\documentclass{sig-alternate-2013}

\usepackage{mathtools,dsfont,subfigure}
\usepackage{color,verbatim,enumitem,psfrag,bbm,comment,tabularx}

\usepackage{appendix}

\usepackage{algorithm,algorithmic}

\usepackage[colorlinks,
            linkcolor=blue,
            citecolor=blue,
            urlcolor=magenta,
            linktocpage,
            plainpages=false,
            pdftex]{hyperref}

\newtheorem{corollary}{Corollary}

\newtheorem{lemma}{Lemma}

\newtheorem{prop}{Proposition}
\newtheorem{theorem}{Theorem} 
 
\newtheorem{definition}{Definition}

\def\EE{{\mathbb{E}}}
\def\PP{{\mathbb{P}}}

\def\setN{\mathbb{N}}
\def\N{\mathcal{N}}

\newcommand{\fall}{\,\forall\,}

\newcommand{\PR}{\pi}
\newcommand{\PRi}{\PR^{-1}}
\newcommand{\e}{\epsilon}
\newcommand{\samp}{k}

\newcommand{\cG}{\mathcal{G}}
\newcommand{\be}{\mathbf{e}}
\newcommand{\bone}{\mathbf{1}}
\newcommand{\signif}{\emph{Detect-High}}
\newcommand{\flo}[1]{\lfloor #1 \rfloor}

\DeclareMathOperator*{\argmax}{arg\,max}

\newfont{\mycrnotice}{ptmr8t at 7pt}
\newfont{\myconfname}{ptmri8t at 7pt}
\let\crnotice\mycrnotice
\let\confname\myconfname

\permission{Permission to make digital or hard copies of all or part of this work for personal or classroom use is granted without fee provided that copies are not made or distributed for profit or commercial advantage and that copies bear this notice and the full citation on the first page. Copyrights for components of this work owned by others than the author(s) must be honored. Abstracting with credit is permitted. To copy otherwise, or republish, to post on servers or to redistribute to lists, requires prior specific permission and/or a fee. Request permissions from permissions@acm.org.}
\conferenceinfo{KDD '14,}{August 24 - 27 2014, New York, NY, USA.\\
{\mycrnotice{Copyright is held by the owner/author(s). Publication rights licensed to ACM.}}}
\copyrightetc{ACM \the\acmcopyr}
\crdata{978-1-4503-2956-9/14/08\ ...\$15.00.\\
http://dx.doi.org/10.1145/2623330.2623745}

\begin{document}

\title{FAST-PPR: Scaling Personalized PageRank Estimation for Large Graphs}

\numberofauthors{4} 

\author{
\alignauthor 
Peter Lofgren\\
        \affaddr{Department of Computer Science}\\
        \affaddr{Stanford University}\\
        \email{plofgren@cs.stanford.edu}
\alignauthor
Siddhartha Banerjee\\
        \affaddr{Department of Management Science \& Engineering}\\
        \affaddr{Stanford University}\\
        \email{sidb@stanford.edu}
\alignauthor
Ashish Goel\\
        \affaddr{Department of Management Science \& Engineering}\\
        \affaddr{Stanford University}\\
        \email{ashishg@stanford.edu}
\and  
\alignauthor 
C. Seshadhri\\
       \affaddr{Sandia National Labs}\\
       \affaddr{Livermore, CA}\\
       \email{scomand@sandia.gov}
}

\maketitle

\begin{abstract}
We propose a new algorithm, FAST-PPR, for estimating personalized
PageRank: given start node $s$ and target node $t$ in a directed graph,
and given a threshold $\delta$, FAST-PPR estimates the Personalized
PageRank $\PR_s(t)$ from $s$ to $t$, guaranteeing a small relative error as long $\PR_s(t)>\delta$. Existing algorithms for this
problem have a running-time of $\Omega(1/\delta)$; in comparison,
FAST-PPR has a provable average running-time guarantee of
${O}(\sqrt{d/\delta})$ (where $d$ is the average in-degree of the graph).
This is a significant improvement, since $\delta$ is often $O(1/n)$
(where $n$ is the number of nodes) for applications. We also complement
the algorithm with an $\Omega(1/\sqrt{\delta})$ lower bound for PageRank
estimation, showing that the dependence on $\delta$ cannot be improved.

We perform a detailed empirical study on numerous massive graphs, showing that FAST-PPR dramatically outperforms existing algorithms. For example, on the 2010 Twitter graph with 1.5 billion edges, for target nodes sampled by popularity, FAST-PPR has a $20$ factor speedup over the state of the art. Furthermore, an enhanced version of FAST-PPR has a $160$ factor speedup on the Twitter graph, and is at least $20$ times faster on all our candidate graphs.
\end{abstract}



\section{Introduction}

The success of modern networks is largely due to the ability to search effectively on them. A key primitive is PageRank \cite{Page1999}, which is widely used as a measure of network importance. The popularity of PageRank is in large part due to its fast computation in large networks. As modern social network applications shift towards being more customized to individuals, there is a need for similar ego-centric measures of network structure. 

\emph{Personalized PageRank} (PPR) \cite{Page1999} has long been viewed as the appropriate ego-centric equivalent of PageRank. For a node $u$, the personalized PageRank vector $\PR_u$ measures the frequency of visiting other nodes via short random-walks from $u$. This makes it an ideal metric for \emph{social search}, giving higher weight to content generated by nearby users in the social graph. Social search protocols find widespread use -- from personalization of general web searches \cite{Page1999,Jeh2003,yin2010}, to more specific applications like collaborative tagging networks \cite{yahia2008}, ranking name search results on social networks \cite{vieira2007}, social Q\&A sites \cite{horowitz2010}, etc. In a typical personalized search application, given a set of candidate results for a query, we want to estimate the Personalized PageRank to each candidate result. This motivates the following problem:
\begin{center}
\emph{Given source node $s$ and target node $t$, estimate the Personalized PageRank $\PR_s(t)$ up to a small relative error.}
\end{center}
Since smaller values of $\PR_s(t)$ are more difficult to detect, we parameterize the problem by threshold $\delta$, requiring small relative errors only if $\PR_s(t) > \delta$. Current techniques used for PPR estimation (see Section \ref{ssec:relwork}) have $\Omega(1/\delta)$ running-time -- this makes them infeasible for large networks when the desired $\delta=O(1/n)$ or $O((\log n)/n)$. 

In addition to social-search, PPR is also used for a variety of other tasks across different domains: friend recommendation on Facebook \cite{backstrom2011supervised}, who to follow on Twitter \cite{gupta2013wtf}, graph partitioning \cite{Andersen2006},  community detection \cite{yang2012defining}, and other applications \cite{tong2006fast}.  Other measures of personalization, such as personalized SALSA and SimRank \cite{sarlos2006randomize}, can be reduced to PPR. However, in spite of a rich body of existing work~\cite{Jeh2003,Avrachenkov2007,Andersen2006,Andersen2007,Bahmani2010,Borgs2013,Sarma2013}, estimating PPR is often a bottleneck in large networks.

\begin{figure*}[t]
\centering
\subfigure[Running time (in log scale) of different algorithms]{

\includegraphics[width=\columnwidth]{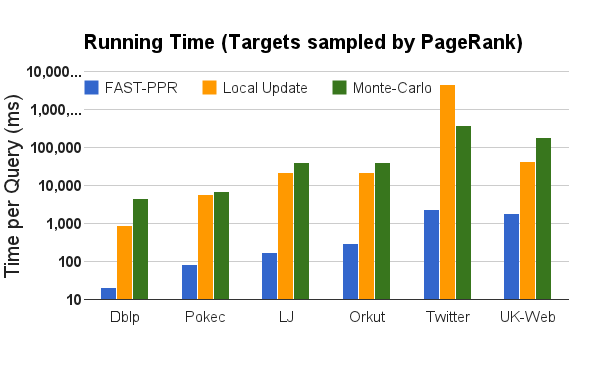}
}
\hfill
\subfigure[Relative accuracy of different algorithms]{
\label{fig:preview_b}
\includegraphics[width=\columnwidth]{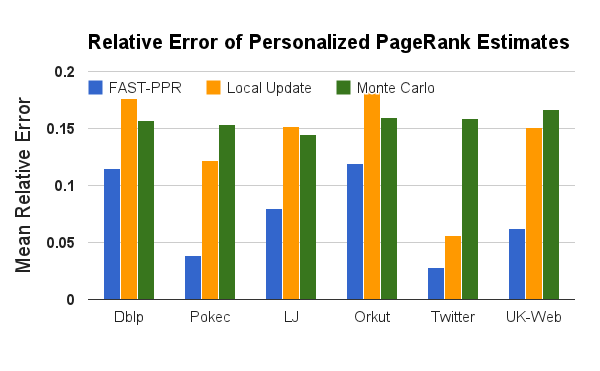}
}

\caption[Comparing different algorithms for PPR]{Comparison of Balanced FAST-PPR, Monte-Carlo and Local-Update algorithms in different networks -- $1000$ source-target pairs, threshold $\delta=\frac{4}{n}$, teleport probability $\alpha=0.2$. Notice that Balanced FAST-PPR is 20 times faster in all graphs, without sacrificing accuracy. For details, see Section \ref{sec:expts}.}
\label{fig:preview}
\end{figure*}

\subsection{Our Contributions}
\label{ssec:contrib}

We develop a new algorithm, \emph{Frontier-Aided Significance Thresholding for Personalized PageRank} (FAST-PPR), based on a new bi-directional search technique for PPR estimation:

\noindent$\bullet$ \textbf{Practical Contributions:} We present a simple implementation of FAST-PPR  which requires no pre-processing and has an average running-time of ${O}(\sqrt{d/\delta})$\footnote{We assume here that the desired relative error and  teleport probability $\alpha$ are constants -- complete scaling details are provided in our Theorem statements in Appendix \ref{appsec:imperfect_alt}.}. We also propose a simple heuristic, Balanced FAST-PPR, that achieves a significant speedup in practice. 

In experiments, FAST-PPR outperforms existing algorithms across a variety of real-life networks. For example, in Figure \ref{fig:preview}, we compare the running-times and accuracies of FAST-PPR with existing methods. Over a variety of data sets, FAST-PPR is significantly faster than the state of the art, with the same or better accuracy.

To give a concrete example: in experiments on the Twitter-2010 graph \cite{WebAlgorithmics}, Balanced FAST-PPR takes less than 3 seconds for random source-target pairs. In contrast, Monte Carlo takes more than 6 minutes and Local Update takes more than an hour. More generally in all graphs, FAST-PPR is at least 20 times faster than the state-of-the-art, without sacrificing accuracy.

\noindent$\bullet$ \textbf{Theoretical Novelty:} FAST-PPR is the first algorithm for PPR estimation with ${O}(\sqrt{d/\delta})$ average running-time, where $d=m/n$ is the average in-degree\footnote{Formally, for $(s,t)$ with $\PR_s(t)>\delta$, FAST-PPR returns an estimate $\widehat{\PR}_s(t)$ with relative error $c$, incurring $O\left(\frac{1}{c^2}\sqrt{\frac{d}{\delta}}\sqrt{\frac{\log(1/p_{fail})\log(1/\delta)}{\alpha^2\log(1/(1-\alpha))}}\right)$ average running-time see Corollary \ref{corr:apprbalance} in Appendix \ref{appsec:imperfect_alt} for details.}. Further, we modify FAST-PPR to get ${O}(1/\sqrt{\delta})$ worst-case running-time, by pre-computing and storing some additional information, with a required storage of ${O}(m/\sqrt{\delta})$.

We also give a new running-time \emph{lower bound} of $\Omega(1/\sqrt{\delta})$ for PPR estimation, which essentially shows that the dependence of FAST-PPR running-time on $\delta$ cannot be improved.

Finally, we note that FAST-PPR has the same performance gains for computing PageRank with arbitrary \emph{preference vectors} \cite{Jeh2003}, where the source is picked from a distribution over nodes. Different preference vectors are used for various applications \cite{Jeh2003,Page1999}. However, for simplicity of presentation, we focus on Personalized PageRank in this work.

\section{Preliminaries}

\label{sec:background}

Given $G(V,E)$, a directed graph, with $|V|=n,|E|=m$, and adjacency matrix $A$. For any node $u\in V$, we denote $\mathcal{N}^{out}(u),d^{out}(u)$ as the out-neighborhood and out-degree respectively; similarly $\mathcal{N}^{in}(u),d^{in}(u)$ are the in-neighborhood and in-degree. We define $d=\frac{m}{n}$ to be the average in-degree (equivalently, average out-degree).

The personalized PageRank vector $\PR_u$ for a node $u\in V$ is the stationary distribution of the following random walk starting from $u$: at each step, return to $u$ with probability $\alpha$, and otherwise move to a random out-neighbor of the current node. Defining $D=\mathrm{diag}(d^{out}(u))$, and $W=D^{-1}A$, the \emph{personalized PageRank} (PPR) vector of $u$ is given by:
\begin{align}
\label{eq:PPR}
\PR_u^T=\alpha \mathbf{e}_u^T+(1-\alpha)\PR_u^T.W,
\end{align}
where $\mathbf{e}_u$ is the identity vector of $u$. Also, for a target node $t$, we define the \emph{inverse-PPR} of a node $w$ with respect to $t$ as $\PR^{-1}_t(w)=\PR_w(t)$. The inverse-PPR vector $\{\PR^{-1}_t(w)\}_{w\in V}$ of $t$ sums to $n\PR(t)$, where $\PR(t)$ is the global PageRank of $t$.

Note that the PPR for a uniform random pair of nodes is $1/n$ -- thus for practical applications, we need to consider $\delta$ of the form $O(1/n)$ or $O(\log n/n)$. We reinforce this choice of $\delta$ using empirical PPR data from the Twitter-2010 graph in Section \ref{ssec:pprdistr} -- in particular, we observe that only 1\

\subsection{Existing Approaches for PPR Estimation}
\label{ssec:relwork}

There are two main techniques used to compute PageRank/PPR vectors. One set of algorithms use power iteration. Since performing a direct power iteration may be infeasible in large networks, a more common approach is to use \emph{local-update} versions of the power method, similar to the Jacobi iteration. This technique was first proposed by Jeh and Widom \cite{Jeh2003}, and subsequently improved by other researchers \cite{Fogaras2005,Andersen2006}. The algorithms are primarily based on the following recurrence relation for $\PR_u$:
\begin{align}
\label{eq:dynamic}
\PR_u^T=\alpha e_u^T+\frac{(1-\alpha)}{|\mathcal{N}^{out}(u)|}.\sum_{v\in\mathcal{N}^{out}(u)}\PR_v^T
\end{align}
Another use of such local update algorithms is for estimating the inverse-PPR vector for a target node. Local-update algorithms for inverse-PageRank are given in \cite{Andersen2007} (where inverse-PPR is referred to as the `contribution PageRank vector'), and \cite{Lofgren2013} (where it is called `susceptibility'). However, one can exhibit graphs where these algorithms need a running-time of $O(1/\delta)$ to get additive guarantees on the order of $\delta$.

Eqn. \ref{eq:dynamic} can be derived from the following probabilistic
re-interpretations of PPR, which also lead to an alternate set of
randomized or \emph{Monte-Carlo} algorithms. Given any random variable
$L$ taking values in $\setN_0$, let
$RW(u,L)\triangleq\{u,V_1,V_2,\ldots,V_L\}$ be a random-walk of random
length $L \sim Geom(\alpha)$\footnote{i.e.,
$\PP[L=i]=\alpha(1-\alpha)^{i}\fall i\in\setN_0$}, starting from $u$.
Then we can write:
\begin{align}
\label{eq:pprmc}
\PR_u(v)&=\PP\left[V_L=v\right]
\end{align}
In other words, $\PR_u(v)$ is the  probability that $v$ is the last node in $RW(u,L)$. Another alternative characterization is:
\[
\PR_u(v)=\alpha\EE\left[\sum_{i=0}^L\mathds{1}_{\{V_i=v\}}\right],
\]
i.e., $\PR_u(v)$ is proportional to the number of times $RW(u,L)$ visits node $v$. Both characterizations can be used to estimate $\PR_u(\cdot)$ via Monte Carlo algorithms, by generating and storing random walks at each node~ \cite{Avrachenkov2007,Bahmani2010,Borgs2013,Sarma2013}. Such estimates are easy to update in dynamic settings \cite{Bahmani2010}. However, for estimating PPR values close to the desired threshold $\delta$, these algorithms need $\Omega(1/\delta)$ random-walk samples. 

\subsection{Intuition for our approach}
\label{sec:prelim}

The problem with the basic Monte Carlo procedure -- generating random walks from $s$ and estimating the distribution of terminal nodes -- is that to estimate a PPR which is $O(\delta)$, we need $\Omega(1/\delta)$ walks. To circumvent this, we introduce a new \emph{bi-directional estimator} for PPR: given a PPR-estimation query with parameters $(s,t,\delta)$, we first work backward from $t$ to find a suitably large set of `targets', and then do random walks from $s$ to test for hitting this set. 

Our algorithm can be best understood through an analogy with the shortest path problem. In the bidirectional shortest path algorithm, to find a path of length $l$ from node $s$ to node $t$, we find all nodes within distance $\frac{l}{2}$ of $t$, find all nodes within distance $\frac{l}{2}$ of $s$, and check if these sets intersect. Similarly, to test if $\PR_s(t) > \delta$, we find all $w$ with $\PR_w(t)> \sqrt{\delta}$ (we call this the \emph{target set}), take $O(1/\sqrt{\delta})$ walks from the start node, and see if these two sets intersect. It turns out that these sets might not intersect even if $\PR_s(t)>\delta$, so we go one step further and consider the \emph{frontier set} -- nodes outside the target set which have an edge into the target set. We can prove that if $\PR_s(t)>\delta$ then random walks are likely to hit the frontier set.

Our method is most easily understood using the characterization of $\PR_s(t)$ as the probability that a single walk from $s$ ends at $t$ (Eqn. \eqref{eq:pprmc}). Consider a random walk from $s$ to $t$ -- at some point, it must enter the frontier. We can then decompose the probability of the walk reaching $t$ into the product of two probabilities: the probability that it reaches some node $w$ in the frontier, and the probability that it reaches $t$ starting from $w$. The two probabilities in this estimate are typically much larger than the overall probability of a random walk reaching $t$ from $s$, so they can be estimated more efficiently. Figure \ref{fig:fastppr_diag} illustrates this bi-directional scheme.

\subsection{Additional Definitions}

To formalize our algorithm, we first need some additional definitions. We define a set $B$ to be a blanket set for $t$ with respect to $s$ if all paths from $s$ to $t$ pass through $B$. Given blanket set $B$ and a random walk $RW(s,L)$, let $H_{B}$ be the first node in $B$ hit by the walk (defining $H_{B}=\bot$ if the walk does not hit $B$ before terminating). Since each walk corresponds to a unique $H_{B}$, we can write $\PP\left[V_L=v\right]$ as a sum of contributions from each node in $B$. Further, from the memoryless property of the geometric random variable, the probability a walk ends at $t$ conditioned on reaching $w\in B$ before stopping is exactly $\PR_s(t)$. Combining these, we have:
\begin{align}
\label{eq:PPRblanket}
\PR_s(t)&=\sum_{w\in B}\PP\left[H_{B}=w \right] \cdot \PR_w(t).
\end{align}
In other words, the PPR from $s$ to $t$ is the sum, over all nodes $w$ in blanket set $B$, of the probability of a random walk hitting $B$ first at node $w$, times the PPR from $w$ to $t$.

Recall we define the \emph{inverse-PPR vector} of $t$ as $\PR^{-1}_t=(\PR_w(t))_{w\in V}$. Now we introduce two crucial definitions:
\begin{definition}[Target Set]
The target set $T_t(\e_r)$ for a target node $t$ is given by:
\begin{equation*}
	 T_t(\e_r):= \{w \in V : \PR^{-1}_t(w)>\e_r\}.
\end{equation*} 
\end{definition}
\begin{definition}[Frontier Set]
The frontier set $F_t(\e_r)$ for a target node $t$ is defined as:
\begin{align*}
	 F_t(\e_r):= \Big(\bigcup_{v \in T_t(\e_r)} \mathcal{N}^{in}_v \Big) \setminus T_t(\e_r).
\end{align*}
\end{definition}
\noindent The target set $T_t(\e_r)$ thus contains all nodes with inverse-PPR greater than $\e_r$, while the frontier set $F_t(\e_r)$ contains all nodes which are in-neighbors of $T_t(\e_r)$, but not in $T_t(\e_r)$.

\subsection{A Bidirectional Estimator}

The next proposition illustrates the core of our approach.
\begin{prop}
\label{prop:frontier}
Set $\e_r < \alpha$.
Fix vertices $s,t$ such that $s\notin T_t(\e_r)$. 
\begin{enumerate}[nolistsep,noitemsep]
\item Frontier set $F_t(\e_r)$ is a blanket set of $t$ with respect to $s$.
\item For random-walk $RW(s,L)$ with length $L\sim Geom(\alpha)$:
	\begin{equation*}
		\PP[RW(s,L)\mbox{ hits }F_t(\e_r)]\geq\frac{\PR_s(t)}{\e_r}
	\end{equation*}
\end{enumerate}
\end{prop}

\begin{proof}
For brevity, let $T_t=T_t(\e_r), F_t=F_t(\e_r)$. By definition, we know $\PR_t(t)\geq\alpha$ -- thus $t\in T_t$, since $\e_r<\alpha$. The frontier set $F_t$ contains all neighbors of nodes in $T_t$ which are not themselves in $T_t$ -- hence, for any source node $u\notin T_t$, a path to $t$ must first hit a node in $F_t$. 

For the second part, since $F_t$ is a blanket set for $s$ with respect to $t$, Eqn. \ref{eq:PPRblanket} implies 
$\PR_s(t)=\sum_{w\in F_t}\PP\left[H_{F_t}=w\right]\PR_w(t)$, where $H_{F_t}$ is the first node where the walk hits $F_t$. Note that by definition, $\forall\, w\in F_t$ we have $w\notin T_t$ -- thus $\PR_w(t)\leq\e_r$. Applying this bound, we get:
\begin{align*}
\PR_s(t)&\leq \e_r\sum_{w\in F_t}\PP\left[H_{F_t}=w\right]=\e_r\PP[RW(s,L)\mbox{ hits }F_t(\e_r)].
\end{align*}	
Rearranging, we get the result.
\end{proof}

\begin{figure}[t]
\centering
\includegraphics[width=0.9\columnwidth]{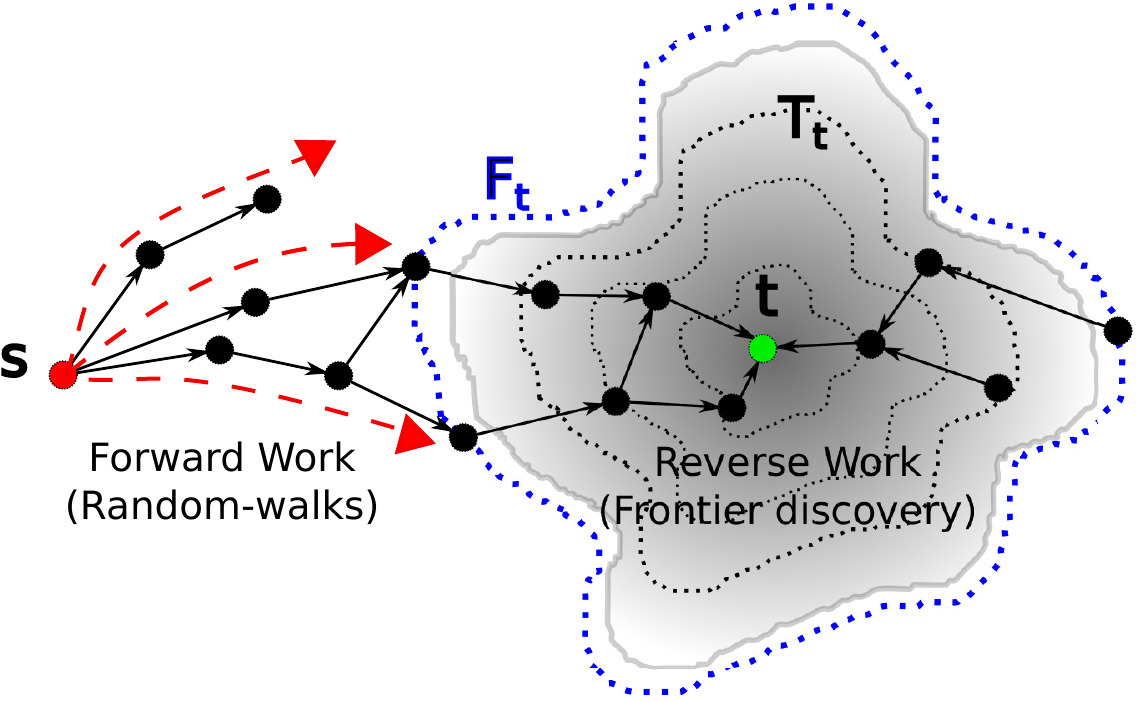}
\caption[FAST-PPR Algorithm]{The FAST-PPR Algorithm: We first work backward from the target $t$ to find the frontier $F_t$ (with inverse-PPR estimates). Next we work forward from the source $s$, generating random-walks and testing for hitting the frontier.}
\label{fig:fastppr_diag}
\end{figure}

The aim is to estimate $\PR_s(t)$ through Eqn.~\ref{eq:PPRblanket}.
By the previous proposition, $F_t(\e_r)$ is a blanket set.
We will determine the set $F_t(\e_r)$ and estimate all quantities
in the right side of Eqn.~\ref{eq:PPRblanket}, thereby estimating $\PR_s(t)$.

We perform a simple heuristic calculation to argue
that setting $\e_r \approx \sqrt{\delta}$ suffices to estimate $\PR_s(t)$. Previous work shows that $T_t(\e_r)$ can be found in $O(d/\e_r)$ time~\cite{Andersen2007,Lofgren2013} -- using this we can find $F_t(\e_r)$. Now suppose that we know all values of $\PR_w(t)$ ($\forall\, w \in F_t(\e_r)$). By Eqn.~\ref{eq:PPRblanket}, we need to estimate the probability of random walks from $s$ hitting vertices in $F_t(\e_r)$. By the previous proposition, the probability of hitting $F_t(\e_r)$ is at least $\delta/\e_r$ -- hence, we need $O(\e_r/\delta)$ walks from $s$ to ensure we hit $F_t(\e_r)$. All in all, we require $O(d/\e_r + \e_r/\delta)$ -- setting $\e_r = \sqrt{d \delta}$ we get a running-time bound of $O(\sqrt{d/\delta})$. In reality, however, we only have (coarse) PPR estimates for nodes in the frontier -- we show how these estimates can be boosted to get the desired guarantees, and also empirically show that, in practice, using the frontier estimates gives good results. Finally, we show that $1/\sqrt{\delta}$ is a fundamental lower bound for this problem.

We note that bi-directional techniques have been used for estimating fixed-length random walk probabilities in \emph{regular undirected graphs} \cite{Goldreich2011,Kale2008}. 

These techniques do not extend to estimating PPR -- in particular, we need to consider \emph{directed graphs, arbitrary node degrees and walks of random length}. Also, Jeh and Widom \cite{Jeh2003} proposed a scheme for PageRank estimation using intermediate estimates from a \emph{fixed skeleton} of target nodes. However there are no running-time guarantees for such schemes; also, the target nodes and partial estimates need to be pre-computed and stored. Our algorithm is fundamentally different as it constructs \emph{separate target sets for each target node} at query-time.

\section{The FAST-PPR Algorithm}
\label{sec:FastPPR}

We now develop the \emph{Frontier-Aided Significance Thresholding} algorithm, or FAST-PPR, specified in Algorithm \ref{alg:FASTPPR}. The input-output behavior of FAST-PPR is as follows:

\noindent$\bullet$ {\bf Inputs:} The primary inputs are graph $G$, teleport probability $\alpha$, start node $s$, target node $t$, and threshold $\delta$ -- for brevity, we henceforth suppress the dependence on $G$ and $\alpha$. We also need a \emph{reverse threshold $\e_r$} -- in subsequent sections, we discuss how this parameter is chosen.

\noindent$\bullet$ {\bf Output:} An estimate $\widehat{\PR}_s(t)$ for ${\PR}_s(t)$.

The algorithm also requires two parameters, $c$ and $\beta$ -- the former controls the number of random walks, while the latter controls the quality of our inverse-PPR estimates in the target set. In our pseudocode (Algorithms \ref{alg:FASTPPR} and \ref{alg:invPPR}), we specify the values we use in our experiments -- the theoretical basis for these choices is provided in Section \ref{ssec:fastpproracle}. 

\begin{algorithm}[ht]
\caption{FAST-PPR$(s,t,\delta)$}
\label{alg:FASTPPR}
\begin{algorithmic}[1] 
\REQUIRE Graph $G$, teleport probability $\alpha$, start node $s$, target node $t$, threshold $\delta$
\STATE Set accuracy parameters $c$, $\beta$ (in our experiments we use $c = 350$, $\beta=1/6$).

\STATE Call FRONTIER$(t,\e_r, \beta)$ to obtain target set $T_t(\e_r)$, frontier set $ F_t(\e_r)$, and inverse-PPR values $(\PR^{-1}_t(w))_{w\in F_t(\e_r) \cup T_t(\e_r)}$.
\IF{$s\in T_t(\e_r)$}
\RETURN $\PR^{-1}_t(s)$
\ELSE
\STATE Set number of walks $\samp=c\e_r/\delta$ \quad(cf. Theorem \ref{thm:fastpprmain})
\FOR{index $i\in [\samp]$}
\STATE Generate $L_i\sim Geom(\alpha)$
\STATE Generate random-walk $RW(s,L_i)$
\STATE Determine $H_i$, the first node in $F_t(\e_r)$ hit by $RW_i$; if $RW_i$ never hits $F_t(\e_r)$, set $H_i=\bot$
\ENDFOR 
\RETURN $\widehat{\PR}_s(t)=(1/\samp)\sum_{i\in[\samp]}\PR^{-1}_t(H_i)$
\ENDIF
\end{algorithmic}
\end{algorithm}

\begin{algorithm}[ht]
\caption{FRONTIER$(t,\e_r, \beta)$~\cite{Andersen2007,Lofgren2013}}
\label{alg:invPPR}
\begin{algorithmic}[1]
\REQUIRE Graph $G$, teleport probability $\alpha$, target node $t$, reverse threshold $\e_r$, accuracy factor $\beta$

\STATE Define additive error $\e_{inv}=\beta\e_r$
\STATE Initialize (sparse) estimate-vector $\widehat{\PR}^{-1}_t$ and (sparse) residual-vector $r_t$ as: 
$  \begin{cases}
   \widehat{\PR}^{-1}_t(u)=r_t(u)=0 \text{ if } u \neq t \\
   \widehat{\PR}^{-1}_t(t)=r_t(t)=\alpha
  \end{cases}$
\STATE Initialize target-set $\widehat{T}_t=\{t\}$, frontier-set $\widehat{F}_t=\{\}$  
\WHILE{$\exists w\in V\,s.t.\,r_t(w)>\alpha\e_{inv}$}
\FOR{$u\in\mathcal{N}^{in}(w)$}
\STATE      $\Delta =  (1-\alpha).\frac{r_t(w)}{d^{out}(u)}$
\STATE      $\widehat{\PR}^{-1}_t(u) = \widehat{\PR}^{-1}_t(u) + \Delta, r_t(u) = r_t(u) + \Delta$
\IF{$\widehat{\PR}^{-1}_t(u)>\e_r$}
\STATE $\widehat{T}_t=\widehat{T}_t \cup \{u\}\,,\, \widehat{F}_t=\widehat{F}_t \cup \mathcal{N}^{in}(u)$
\ENDIF
\ENDFOR
\STATE  Update $r_t(w)=0$
\ENDWHILE
\STATE $\widehat{F}_t=\widehat{F}_t\setminus \widehat{T}_t$

\RETURN $\widehat{T}_t,\widehat{F}_t, (\widehat{\PR}^{-1}_t(w))_{w\in \widehat{F}_t \cup \widehat{T}_t}$
\end{algorithmic}
\end{algorithm}    

FAST-PPR needs to know sets $T_t(\e_r),F_t(\e_r)$ and inverse-PPR values $(\PR^{-1}_t(w))_{w\in F_t(\e_r) \cup T_t(\e_r)}$. These can be obtained (approximately) from existing algorithms of \cite{Andersen2007,Lofgren2013}. For the sake of completeness, we provide pseudocode for the procedure FRONTIER (Algorithm~\ref{alg:invPPR}) that obtains this information. The following combines Theorems 1 and 2 in~\cite{Lofgren2013}.

\begin{theorem} 
	\label{thm:front} 
FRONTIER$(t,\e_r,\beta)$ algorithm computes estimates $\widehat{\PR}^{-1}_t(w)$ for every vertex $w$, with a guarantee that $\forall w$, $|\widehat{\PR}^{-1}_t(w) - {\PR}^{-1}_t(w)| < \beta \e_r$.
The average running time (over all choices of $t$) is $O(d/(\alpha\e_r))$, where $d = m/n$ is the average degree of the graph.
\end{theorem}

Observe that the estimates $\widehat{\PR}^{-1}_t(w)$ are used to find approximate
target and frontier sets. Note that the running time for a given $t$ is proportional
to the frontier size $|\widehat{F}_t|$. It is relatively straightforward to argue (as in~\cite{Lofgren2013})
that $\sum_{t \in V} |\widehat{F}_t| = \Theta(nd/(\alpha\e_r))$.

In the subsequent subsections, we present theoretical analyses of the running times
and correctness of FAST-PPR. The correctness proof makes an excessively strong assumption
of perfect outputs for FRONTIER, which is not true. To handle this problem, we have a more
complex variant of FAST-PPR that can be proven theoretically (see Appendix \ref{appsec:imperfect_alt}). Nonetheless, our empirical
results show that FAST-PPR does an excellent job of estimate $\PR_s(t)$.

\subsection{Running-time of FAST-PPR}

\begin{theorem}
\label{thm:perffastppr}
Given parameters $\delta, \e_r$, the running-time of the FAST-PPR algorithm, averaged over uniform-random pairs $s,t$, is $O(\alpha^{-1} (d/\e_r + \e_r/\delta))$.

\end{theorem}

\begin{proof}

Each random walk $RW(u,Geom(\alpha))$ takes $1/\alpha$ steps on average,
and there are $O(\e_r/\delta)$ such walks performed -- this is the \emph{forward time}. On the other hand, from Theorem \ref{thm:front}, we have that for a random $(s,t)$ pair, the average running time of FRONTIER is $O(d/(\alpha \e_r))$ -- this is the \emph{reverse time}. Combining the two, we get the result.
\end{proof}
Note that the reverse time bound above is averaged across choice of target node; for some target nodes (those with high global PageRank) the reverse time may be much larger than average, while for others it may be smaller.  However, the forward time is similar for all source nodes, and is predictable -- we exploit this in Section \ref{ssec:rebalancing} to design a balanced version of FAST-PPR which is much faster in practice. In terms of theoretical bounds, the above result suggests an obvious choice of $\e_r$ to optimize the running time:

\begin{corollary}
\label{corr:balance}
Set $\e_r=\sqrt{d\delta}$. Then FAST-PPR has an average per-query running-time of
$O(\alpha^{-1}\sqrt{d/\delta})$.
\end{corollary}
\vfill\eject 
\subsection{FAST-PPR with Perfect FRONTIER}
\label{ssec:fastpproracle}

We now analyze FAST-PPR in an idealized setting, where we assume that FRONTIER returns \emph{exact inverse-PPR estimates} -- i.e., the sets $T_t(\e_r)$, $F_t(\e_r)$, and the values $\{\PR^{-1}_t(w)\}$ are known exactly. This is an unrealistic assumption, but it gives much intuition
into \emph{why} FAST-PPR works. In particular, we show that if $\pi(s,t)> \delta$, then with probability at least 99\
\begin{theorem}
\label{thm:fastpprmain} For any $s,t,\delta,\e_r$,
FAST-PPR outputs an estimate $\widehat{\PR}_s(t)$ such
that with probability $>0.99$:
$$|\PR_s(t)-\widehat{\PR}_s(t)| \leq \max(\delta,\PR_s(t))/4.$$
\end{theorem}

\begin{proof} 
We choose $c=\max \left(48 \cdot 8 e \ln(100), 4 \log_2(100)\right)$; this choice of parameter $c$ is for ease of exposition in the computations below, and has not been optimized. 

To prove the result, note that FAST-PPR performs $\samp = c\e_r/\delta$
i.i.d. random walks $RW(s,L)$. We use $H_i$ to denote the first node in
$F_t(\e_r)$ hit by the $i$th random walk. Let $X_i = \PR^{-1}_t(H_i)$ and
$X = \sum_{i=1}^\samp X_i$. By Eqn. \ref{eq:PPRblanket}, $\EE[X_i] =
\PR_s(t)$, so $\EE[X] = \samp \PR_s(t)$. Note that $\widehat{\PR}_s(t) =
X/\samp$. As result, $|\PR_s(t)-\widehat{\PR}_s(t)|$ is exactly $|X -
\EE[X]|/\samp$.

It is convenient to define scaled random variables $Y_i = X_i/\e_r$
and $Y = X/\e_r$ before we apply standard Chernoff bounds.  
We have $|\PR_s(t)-\widehat{\PR}_s(t)| = (\e_r/\samp)|Y - \EE[Y]|$
$ = (\delta/c)|Y - \EE[Y]|$. Also, $\PR_s(t) = (\delta/c)\EE[Y]$.
Crucially, because  $H_i \in F_t(\e_r)$, $X_i=\PR_{H_i}(t) < \e_r$, so   $Y_i \leq 1$.
Hence, we can apply the following two Chernoff bounds (refer to Theorem 1.1 in~\cite{DuPa09}):
\begin{enumerate}
\item $\PP[|Y - \EE[Y]| > \EE[Y]/4] < \exp(-\EE[Y]/48)$
\item $\textrm{For any } b > 2e\EE[Y], \PP[Y > b] \leq 2^{-b}$
\end{enumerate}

Now, we perform a case analysis. Suppose $\PR_s(t) > \delta/(4e)$.
Then $\EE[Y] > c/(4e)$, and
\begin{eqnarray*}
	\PP[|\PR_s(t)-\widehat{\PR}_s(t)| > \PR_s(t)/4] & = & \PP[|Y - \EE[Y]| > \EE[Y]/4] \\
	& < & \exp(-c/48\cdot 4e) < 0.01
\end{eqnarray*}
Suppose $\PR_s(t) \leq \delta/(4e)$.
Then, $\delta/4 > 2e \PR_s(t)$ implying $c/4 > 2e \EE[Y]$. By the upper tail:

\begin{eqnarray*}
	\PP[\widehat{\PR}_s(t) > \delta/4] = \PP[Y > c/4] \leq 2^{-c/4} < 0.01
\end{eqnarray*}

The proof is completed by trivially combining both cases.
\end{proof}

\section{Lower bound for PPR Estimation}
\label{sec:lower}

In this section, we prove that any algorithm that accurately estimates PPR queries up to a threshold $\delta$ must look at $\Omega(1/\sqrt{\delta})$ edges of the graph. 
Thus, our algorithms have the optimal dependence on $\delta$. The numerical constants below are chosen for easier calculations, and are not optimized. 

We assume $\alpha = 1/100\log(1/\delta)$, and consider randomized algorithms for the following variant of Significant-PPR, which we denote as \signif$(\delta)$ -- for all pairs $(s,t)$:
\begin{itemize}[nolistsep,noitemsep]
\item If $\PR_s(t) > \delta$, output ACCEPT with probability $> 9/10$.
\item If $\PR_s(t) < \frac{\delta}{2}$, output REJECT with probability $>9/10$.
\end{itemize}
We stress that the probability is over the random choices of the algorithm, \emph{not} over $s,t$. We now have the following lower bound:
\begin{theorem} 
\label{thm:lb} 
Any algorithm for \signif$(\delta)$ must access $\Omega(1/\sqrt{\delta})$ edges of the graph.
\end{theorem}

\begin{proof}[Outline]
The proof uses a lower bound of Goldreich and Ron for \emph{expansion testing}~\cite{GoRo97}. The technical content of this result is the following -- consider two distributions $\cG_1$ and $\cG_2$ of undirected $3$-regular graphs on $N$ nodes. A graph in $\cG_1$ is generated
by choosing three uniform-random perfect matchings of the nodes. A graph in $\cG_2$ is generated by randomly partitioning the nodes into $4$ equally sized sets, $V_i, i \in \{1,2,3,4\}$, and then, within each $V_i$, choosing three uniform-random matchings. 

Consider the problem of distinguishing $\cG_1$ from $\cG_2$. An adversary arbitrarily picks one of these distributions, and generates a graph $G$ from it. A \emph{distinguisher} must report whether $G$ came from $\cG_1$ or $\cG_2$, and it must be correct with probability $> 2/3$ regardless of the distribution chosen. Theorem 7.5 of Goldreich and Ron~\cite{GoRo97} asserts the following:
\begin{theorem}[Theorem 7.5 of \cite{GoRo97}] 
\label{thm:GR} 
Any distinguisher must look at $\sqrt{N}/5$ edges of the graph.
\end{theorem}
We perform a direct reduction, relating the \signif$(\delta)$ problem to the Goldreich and Ron setting -- in particular, we show that an algorithm for \signif$(\delta)$ which requires less than $1/\sqrt{\delta}$ queries can be used to construct a distinguisher which violates Theorem \ref{thm:GR}. The complete proof is provided in Appendix \ref{appsec:lower}.
\end{proof}

\section{Further variants of FAST-PPR} \label{sec:variant}

The previous section describes vanilla FAST-PPR, with a proof of
correctness assuming a perfect FRONTIER. We now present some variants of
FAST-PPR. We give a theoretical variant that is a truly provable
algorithm, with no assumptions required -- however, vanilla FAST-PPR is
a much better practical candidate and it is what we implement. We also discuss how we can use pre-computation and storage to obtain worst-case guarantees for FAST-PPR. Finally, from the practical side, we discuss a workload-balancing heuristic that provides significant improvements 
in running-time by dynamically adjusting $\e_r$.

\subsection{Using Approximate Frontier-Estimates}
\label{ssec:fastpprapprox}

The assumption that FRONTIER returns perfect estimates is theoretically
untenable -- Theorem~\ref{thm:front} only ensures that each inverse-PPR
estimate is correct up to an additive factor of $\e_{inv}=\beta \e_r$. It
is plausible that for every $w$ in the frontier $\widehat{F}_t(\e_r)$,
$\PR^{-1}_t(w) < \e_{inv}$, and FRONTIER may return a zero estimate for
these PPR values. It is not clear how to use these noisy
inverse-PPR estimates to get the desired accuracy guarantees.

To circumvent this problem, we observe that estimates $\PR^{-1}_t(w)$
for any node $w \in \widehat{T}_t(\e_r)$ \emph{are in fact accurate up to a multiplicative factor}. We design a procedure to bootstrap these `good' estimates, by using a special `target-avoiding' random walk -- this modified algorithm gives the desired accuracy guarantee with only an additional $\log(1/\delta)$ factor in running-time. The final algorithm and proof are quite intricate -- due to lack of space, the details are deferred to  Appendix \ref{appsec:imperfect_alt}.

\subsection{Balanced FAST-PPR}
\label{ssec:rebalancing}

In FAST-PPR, the parameter $\epsilon_r$ can be chosen freely while preserving accuracy.  Choosing a larger value leads to a smaller frontier and less forward work at the cost of more reverse work; a smaller value requires more reverse work and fewer random walks. To improve performance in practice, we can optimize the choice of $\epsilon_r$ based on the target $t$, to balance the reverse and forward time. Note that for any value of $\epsilon_r$, the forward time is proportional to $k=c\e_r/\delta$ (the number of random walks performed) -- for any choice of $\e_r$, it is easy to estimate the forward time required.
Thus, instead of committing to a single value of $\epsilon_r$, we propose a heuristic, BALANCED-FAST-PPR, wherein we \emph{dynamically decrease $\epsilon_r$  until the estimated remaining forward time equals the reverse time already spent}.  

We now describe BALANCED-FAST-PPR in brief: Instead of pushing from any node $w$ with residual $r_t(w)$ above a fixed threshold $\alpha \epsilon_{inv}$, we now push from the node $w$ with the largest residual value -- this follows a similar algorithm proposed in \cite{Lofgren2013}. From \cite{Andersen2007,Lofgren2013}, we know that a current maximum residual value of $r_{max}$ implies an additive error guarantee of $\frac{r_{max}}{\alpha}$ -- this corresponds to a dynamic $\epsilon_r$ value of $\frac{r_{max}}{\alpha \beta}$. At this value of $\epsilon_r$, the number of forward walks required is $k=c\e_r/\delta$. By multiplying $k$ by the average time needed to generate a walk, we get a good prediction of the amount of forward work still needed -- we can then compare it to the time already spent on reverse-work and adjust $\epsilon_r$ until they are equal.  Thus BALANCED-FAST-PPR is able to dynamically choose $\epsilon_r$ to balance the forward and reverse running-time. Complete pseudocode is in Appendix \ref{appsec:balanced_code}.  In Section \ref{ssec:balanceexpt}, we experimentally show how this change balances forward and reverse running-time, and significantly reduces the average running-time.  

\subsection{FAST-PPR using Stored Oracles}
\label{ssec:storage}

All our results for FAST-PPR have involved average-case running-time bounds. To convert these to worst-case running-time bounds, we can pre-compute and store the frontier for all nodes, and only perform random walks at query time. To obtain the corresponding storage requirement for these \emph{Frontier oracles}, observe that for any node $w\in V$, it can belong to the target set of at most $\frac{1}{\e_r}$ nodes, as $\sum_{t\in V}\PR^{-1}_t(w)=1$. Summing over all nodes, we have:
\begin{align*}
\mbox{Total Storage}&\leq\sum_{t\in V}\sum_{w\in T_t}\sum_{u\in\mathcal{N}^{in}(w)}\mathds{1}_{\{u\in F_t\}}\\
&\leq\sum_{w\in V}\sum_{t\in V:w\in T_t}d^{in}(w)\leq\frac{m}{\e_r}
\end{align*}

To further cut down on running-time, we can also pre-compute and store the random-walks from all nodes, and perform appropriate joins at query time to get the FAST-PPR estimate. This allows us to implement FAST-PPR on any distributed system that can do fast intersections/joins. More generally, it demonstrates how the modularity of FAST-PPR can be used to get variants that trade-off between different resources in practical implementations.

\section{Experiments}
\label{sec:expts}

\noindent We conduct experiments to explore three main questions:
\begin{enumerate}[nolistsep,noitemsep]
\item How fast is FAST-PPR relative to previous algorithms?
\item How accurate are FAST-PPR's estimates?
\item How is FAST-PPR's performance affected by our design choices: use of frontier and balancing forward/reverse running-time?  
\end{enumerate}

\subsection{Experimental Setup}
\label{ssec:exptsetup}

\noindent\textbf{$\bullet$ Data-Sets:} To measure the robustness of FAST-PPR, we run our experiments on several types and sizes of graph, as described in Table \ref{table:graphs}.  
\begin{table}
\centering
\caption{Datasets used in experiments}
\label{table:graphs}
\begin{tabular}{|c|c|c|c|} \hline
Dataset & Type & \# Nodes & \# Edges\\ \hline
DBLP-2011 & undirected & 1.0M & 6.7M\\ \hline
Pokec & directed & 1.6M & 30.6M\\ \hline
LiveJournal & undirected & 4.8M & 69M\\ \hline
Orkut & undirected & 3.1M & 117M\\ \hline
Twitter-2010 & directed & 42M & 1.5B\\ \hline
UK-2007-05 & directed & 106M & 3.7B\\ \hline
\end{tabular}
\end{table}

Pokec and Twitter are both social networks in which edges are directed.  
The LiveJournal, Orkut (social networks) and DBLP (collaborations on papers) networks are all undirected -- for each, we have the largest connected component of the overall graph. Finally, our largest dataset with $3.7$ billion edges is from a 2007 crawl of the UK domain \cite{BRSLLP, BSVLTAG}.  Each vertex is a web page and each edge is a hyperlink between pages. 

For detailed studies of FAST-PPR, we use the Twitter-2010 graph, with $41$ million users and $1.5$ billion edges. This presents a further algorithmic challenge because of the skew of its degree distribution: the average degree is $35$, but one node has more than $700,000$ in-neighbors.  

The Pokec \cite{takac2012data}, Live Journal \cite{mislove-2007-socialnetworks}, and Orkut \cite{mislove-2007-socialnetworks} datasets were downloaded from the Stanford SNAP project \cite{SnapProject}. The  DBLP-2011 \cite{BRSLLP}, Twitter-2010 \cite{BRSLLP} and UK 2007-05 Web Graph \cite{BRSLLP, BSVLTAG} were downloaded from the Laboratory for Web Algorithmics \cite{WebAlgorithmics}.  

\noindent\textbf{$\bullet$ Implementation Details:} We ran our experiments on a machine with a 3.33 GHz 12-core Intel Xeon X5680 processor, 12MB cache, and 192 GB of 1066 MHz Registered ECC DDR3 RAM.  
Each experiment ran on a single core and loaded the graph used into memory before beginning any timings.  The RAM used by the experiments was dominated by the RAM needed to store the largest graph using the SNAP library format \cite{SnapProject}, which was about 21GB.

For reproducibility, our C++ source code is available at:
\url{http://cs.stanford.edu/~plofgren/fast_ppr/}

\noindent\textbf{$\bullet$ Benchmarks:} We compare FAST-PPR to two benchmark algorithms: Monte-Carlo and Local-Update. 

Monte-Carlo refers to the standard random-walk algorithm~\cite{Avrachenkov2007,Bahmani2010,Borgs2013,Sarma2013} -- we perform $\frac{c_{MC}}{\delta}$ walks and estimate $\PR_u(v)$ by the fraction of walks terminating at $v$. For our experiments, we choose $c_{MC}=35$, to ensure that the relative errors for Monte-Carlo are the same as the relative error bounds chosen for Local-Update and FAST-PPR (see below). However, even in experiments with $c_{MC}=1$, we find that FAST-PPR is still 3 times faster on all graphs and 25 times faster on the largest two graphs (see Appendix \ref{appsec:expts}).

Our other benchmark, Local-Update, is the state-of-the-art local power iteration algorithm~\cite{Andersen2007,Lofgren2013}. It follows the same procedure as the FRONTIER algorithm (Algorithm \ref{alg:invPPR}), but with the additive accuracy $\e_{inv}$ set to $\delta/2$. Note that a backward local-update is more suited to computing PPR forward schemes~\cite{Andersen2006,Jeh2003} as the latter lack natural performance guarantees on graphs with high-degree nodes.

\noindent\textbf{$\bullet$ Parameters:} For FAST-PPR, we set the constants $c = 350$ and $\beta = 1/6$ -- these are guided by the Chernoff bounds we use in the proof of Theorem \ref{thm:fastpprmain}. For vanilla FAST-PPR, we simply choose $\e_r = \sqrt{\delta}$.

\subsection{Distribution of PPR values}
\label{ssec:pprdistr}

For all our experiments, we use $\delta=\frac{4}{n}$. To understand the importance of this threshold, we study the distribution of PPR values in real networks. Using the Twitter graph as an example, we choose 10,000 random $(s,t)$ pairs and compute $\PR_s(t)$ using FAST-PPR to accuracy $\delta=\frac{n}{10}$. The complementary cumulative distribution function is shown on a log-log plot in Figure \ref{fig:ppr_ccdf}.  Notice that the plot is roughly linear, suggesting a power-law. Because of this skewed distribution, only 2.8\

\begin{figure}
\centering
\includegraphics[width=\columnwidth]{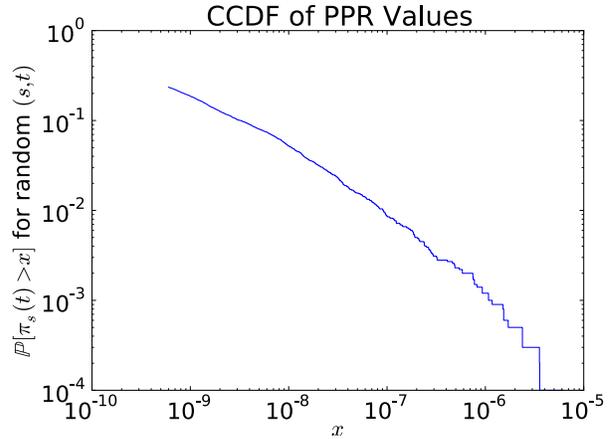}
\caption{Complementary cumulative distribution for 10,000 (s,t) pairs sampled uniformly at random on the Twitter graph.} 
\label{fig:ppr_ccdf}
\end{figure}

\subsection{Running Time Comparisons}
\label{ssec:timeexpts}

After loading the graph into memory, we sample $1000$ source/target pairs $(s,t)$ uniformly at random. For each, we measure the time required for answering PPR-estimation queries with threshold $\delta=4/n$, which, as we discuss above, is fairly significant because of the skew of the PPR distribution. To keep the experiment length less than 24 hours, for the Local-Update algorithm and Monte-Carlo algorithms we only use 20 and 5 pairs respectively.

\begin{figure*}[t]
\centering
\subfigure[Sampling targets uniformly]{
\label{fig:runtime_0_2_a}
\includegraphics[width=1\columnwidth]{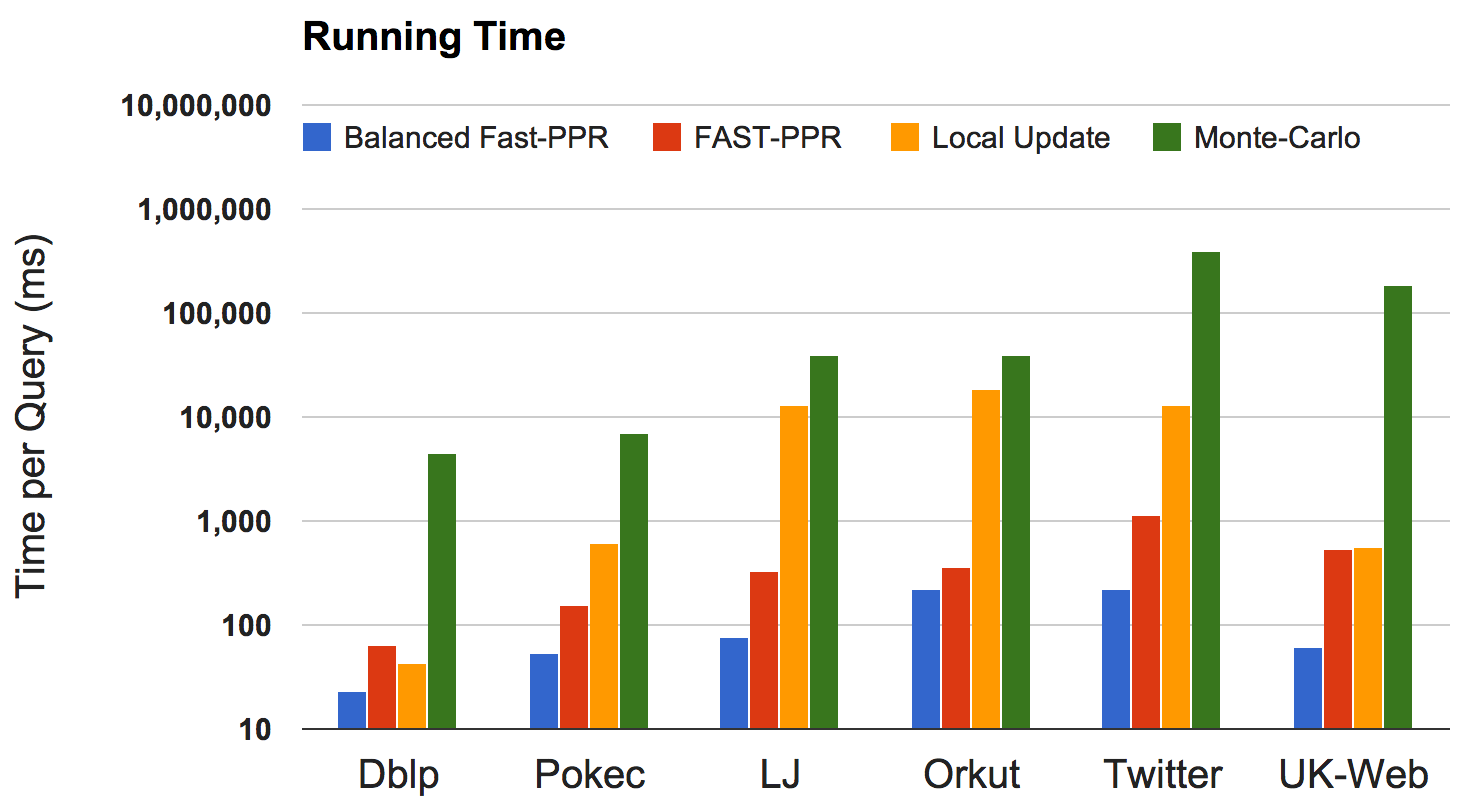}
}
\hfill
\subfigure[Sampling targets from PageRank distribution]{
\label{fig:runtime_0_2_b}
\includegraphics[width=1\columnwidth]{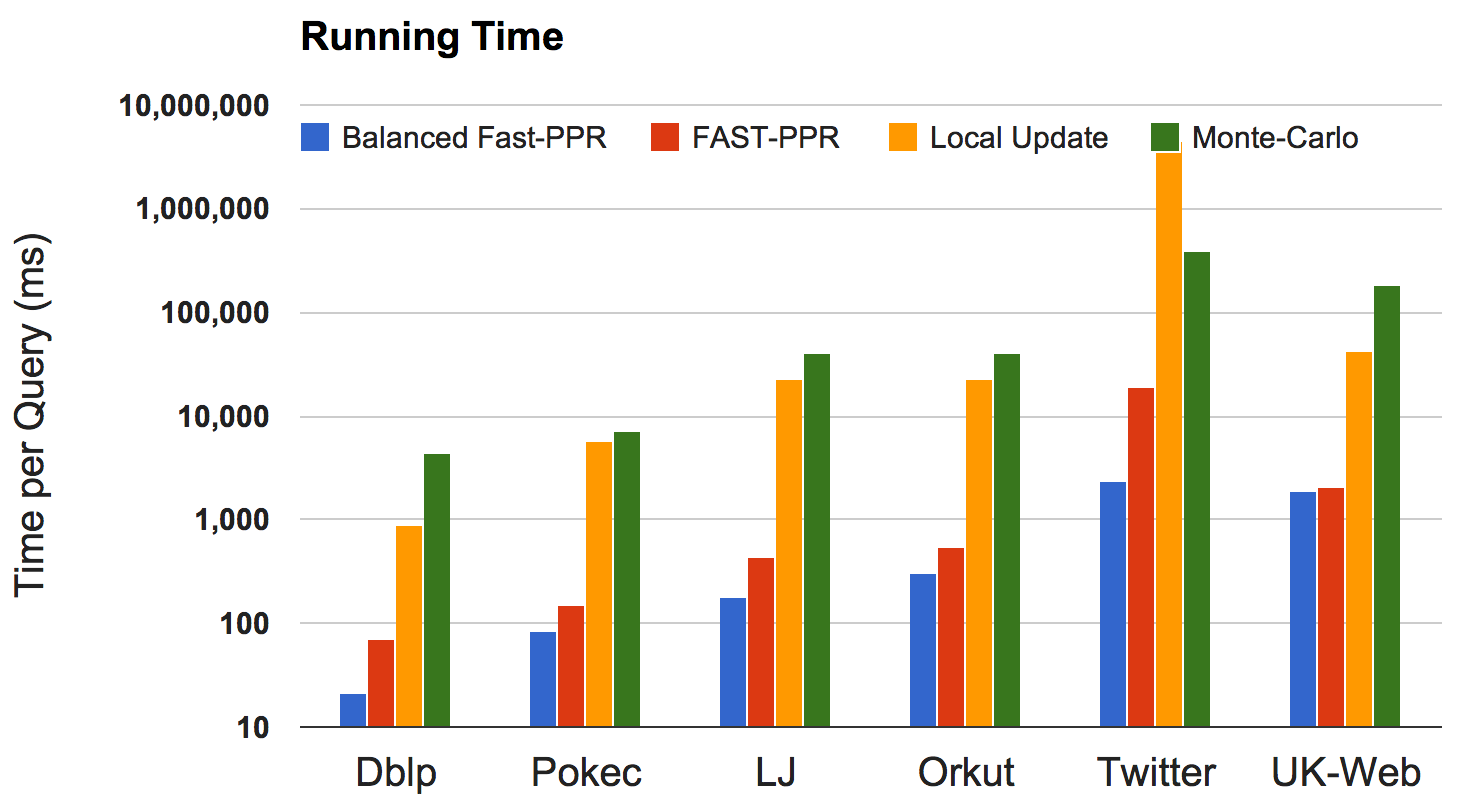}
}

\caption[Running-time Plots]{Average running-time (on log-scale) for different networks. We measure the time required for Significant-PPR queries $(s,t,\delta)$ with threshold $\delta=\frac{4}{n}$ for $1000$ $(s,t)$ pairs. For each pair, the start node is sampled uniformly, while the target node is sampled uniformly in Figure \ref{fig:runtime_0_2_a}, or from the global PageRank distribution in Figure \ref{fig:runtime_0_2_b}. In this plot we use teleport probability $\alpha=0.2$.}
\label{fig:runtime_0_2}
\end{figure*}

The running-time comparisons are shown in Figure \ref{fig:runtime_0_2} --
we compare Monte-Carlo, Local-Update, vanilla FAST-PPR, and BALANCED-FAST-PPR. We perform an analogous experiment where target nodes are
sampled according to their global PageRank value. This is a more
realistic model for queries in personalized search applications, with
searches biased towards more popular targets. The results, plotted in
Figure~\ref{fig:runtime_0_2_b}, show even better speedups for FAST-PPR. All in all, FAST-PPR is many orders of magnitude faster than the state of the art.

{\bf The effect of target global PageRank:} To quantify the speedup further, we sort the targets in the Twitter-2010 graph by their global PageRank, and choose the first target in each percentile.
We measure the running-time of the four algorithms (averaging over random source nodes), as shown in Figure \ref{fig:runtime_pagerank}. 
Note that FAST-PPR is much faster than previous methods for the targets with high PageRank. Note also that large PageRank targets account for most of the average running-time -- thus improving performance in these cases causes significant speedups in the average running time.
We also see that BALANCED-FAST-PPR has significant improvements over vanilla FAST-PPR, especially for lower PageRank targets.

\begin{figure}[!t]
\centering

\includegraphics[width=\columnwidth]{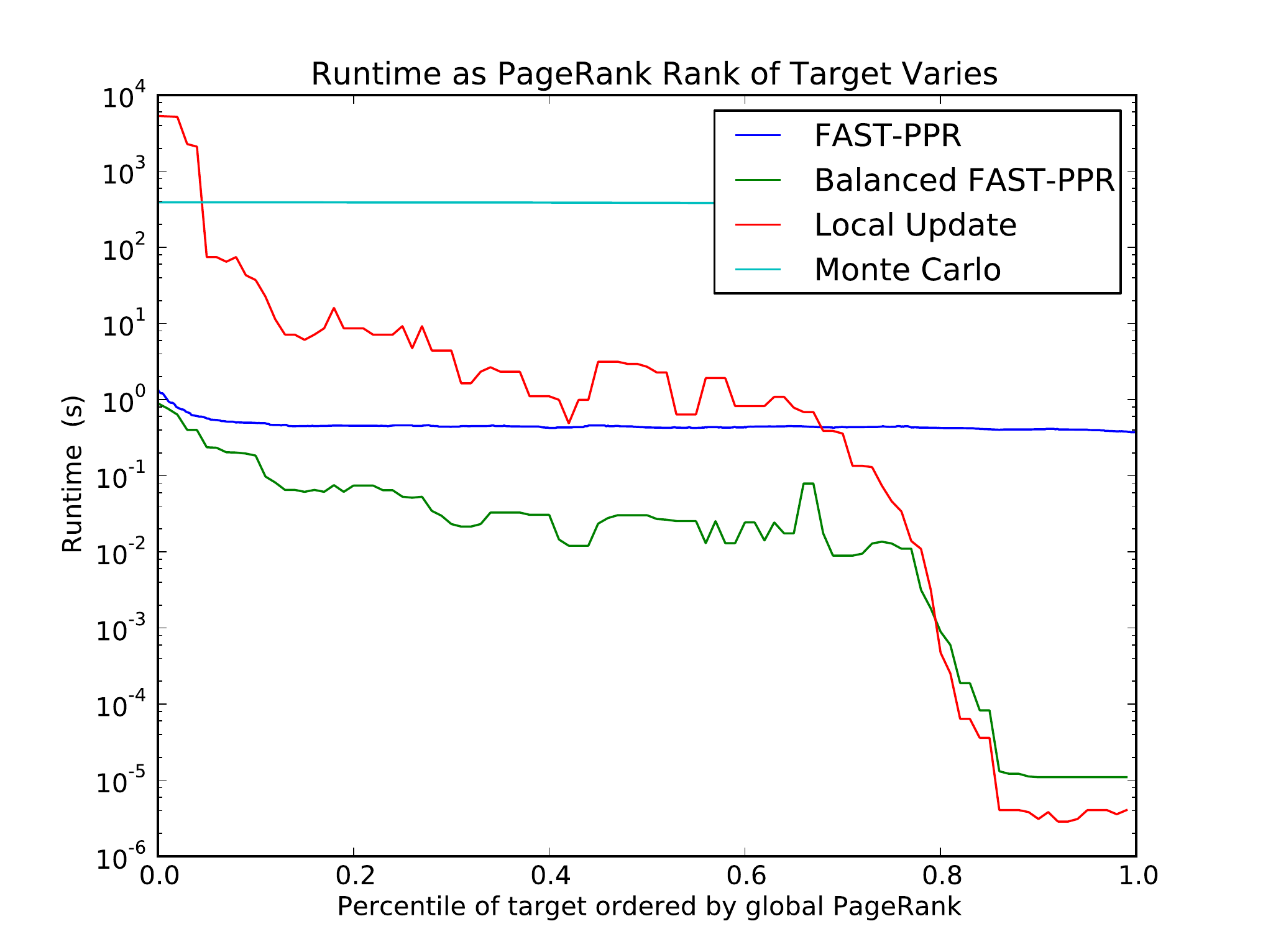}
\caption{Execution time vs. global PageRank of the target. Target nodes were sorted by global PageRank, then one was chosen from each percentile. We use $\alpha=0.2$, and $5$-point median smoothing.}
\label{fig:runtime_pagerank}
\end{figure}

To give a sense of the speedup achieved by FAST-PPR, consider the Twitter-2010 graph. Balanced FAST-PPR takes less than 3 seconds for Significant-PPR queries with targets sampled from global PageRank -- in contrast, Monte Carlo takes more than 6 minutes and Local Update takes more than an hour. In the worst-case, Monte Carlo takes 6 minutes and Local Update takes 6 hours, while Balanced FAST-PPR takes 40 seconds. Finally, the estimates from FAST-PPR are twice as accurate as those from Local Update, and 6 times more accurate than those from Monte Carlo.

\subsection{Measuring the Accuracy of FAST-PPR}
\label{ssec:exptacc}

We measure the empirical accuracy of FAST-PPR. For each graph, we sample $25$ targets uniformly at random, and compute their ground truth inverse-PPR vectors by running a power iteration up to an additive error of $\delta/100$ (as before, we use $\delta=4/n$). Since larger PPR values are easier to compute than smaller PPR values, we sample start nodes such that $\PR_s(t)$ is near the significance threshold $\delta$. In particular, for each of the 25 targets $t$, we sample 50 random nodes from the set $\{s: \delta/4 \leq \PR_s(t) \leq \delta\}$ and 50 random nodes from the set  $\{s: \delta \leq \PR_s(t) \leq 4 \delta\}$. 

We execute FAST-PPR for each of the 2500 $(s,t)$ pairs, and measure the empirical error -- the results are compiled in Table \ref{table:accuracy}. Notice that FAST-PPR has mean relative error less than 15\

\begin{table*}
\centering
\caption{Accuracy of FAST-PPR (with parameters as specified in Section \ref{ssec:exptsetup})}
\label{table:accuracy}
\begin{tabular}{|c|c|c|c|c|c|c|} 
\hline
{} & Dblp & Pokec & LJ & Orkut & Twitter & UK-Web\\ \hline
Threshold $\delta$ & 4.06e-06 & 2.45e-06 & 8.25e-07 & 1.30e-06 & 9.60e-08 & 3.78e-08\\ \hline
Average Additive Error & 5.8e-07 & 1.1e-07 & 7.8e-08 & 1.9e-07 & 2.7e-09 & 2.2e-09\\ \hline
Max Additive Error & 4.5e-06 & 1.3e-06 & 6.9e-07 & 1.9e-06 & 2.1e-08 & 1.8e-08\\ \hline
Average Relative Error: & 0.11 & 0.039 & 0.08 & 0.12 & 0.028 & 0.062\\ \hline
Max Relative Error & 0.41 & 0.22 & 0.47 & 0.65 & 0.23 & 0.26\\ \hline

\end{tabular}
\end{table*}

To make sure that FAST-PPR is not sacrificing accuracy for improved running-time, we also compute the relative error of  Local-Update and Monte-Carlo, using the same parameters as for our running-time experiments. For each of the 2500 $(s,t)$ pairs, we run Local-Update, and measure its relative error. For testing Monte-Carlo, we use our knowledge of the ground truth PPR, and the fact that each random-walk from $s$ terminates at $t$ with probability $p_s=\PR_{s}(t)$. This allows us to simulate Monte-Carlo by directly sampling from a Bernoulli variable with mean $\PR_{s}(t)$ -- this statistically identical to generating random-walks and testing over all pairs.  Note that actually simulating the walks would take more than 50 days of computation for 2500 pairs. The relative errors are shown in Figure \ref{fig:preview_b}.  Notice that FAST-PPR is more accurate than the state-of-the-art competition on all graphs.  This shows that our running time comparisons are using parameters settings that are fair.

\begin{figure*}[t]
\centering

\subfigure[FAST-PPR]{
\label{fig:accuracy_c}
\includegraphics[width=0.85\columnwidth]{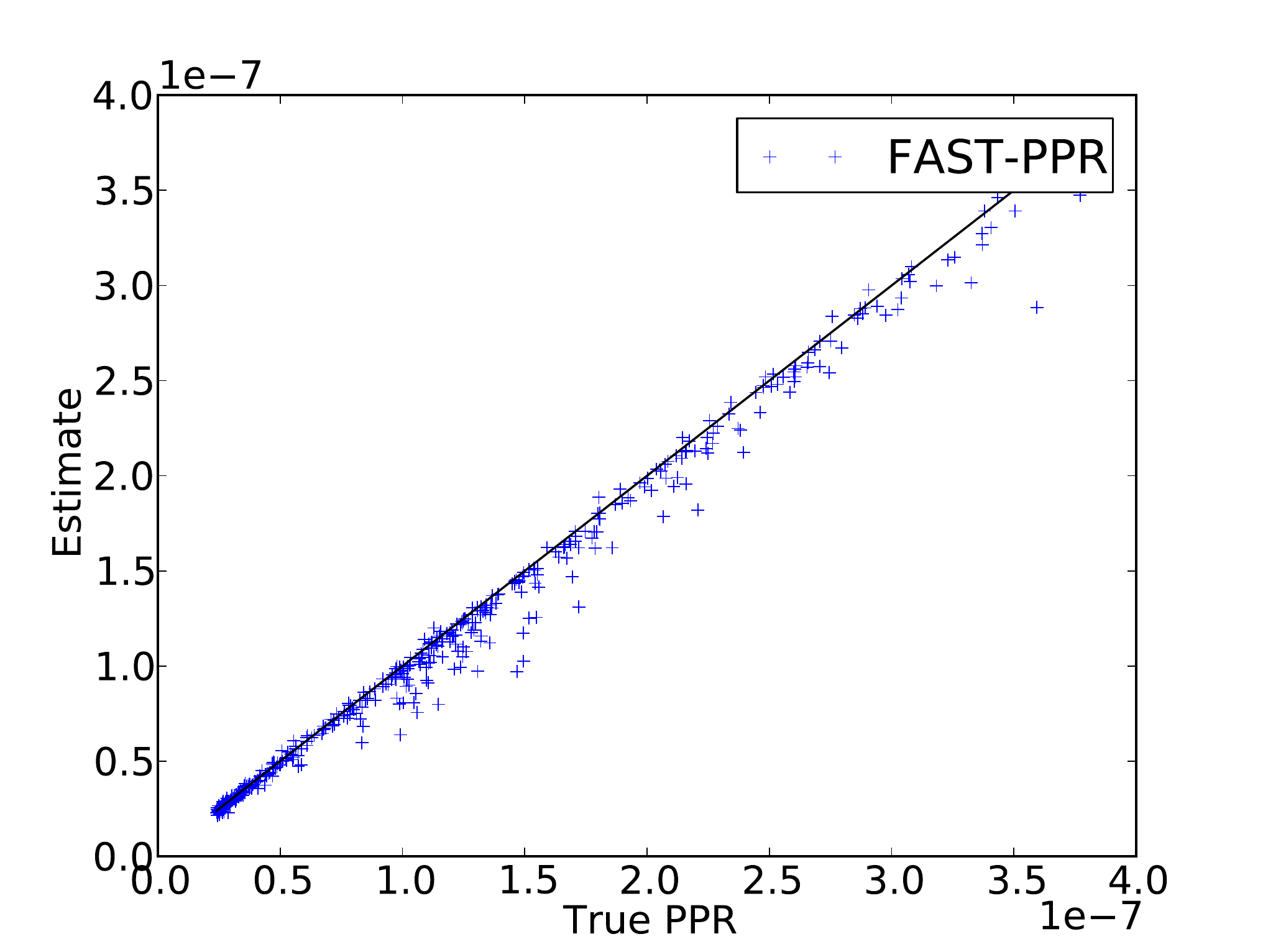}
}
\hspace{1cm}
\subfigure[FAST-PPR using target set instead of frontier]{
\label{fig:accuracy_d}
\includegraphics[width=0.85\columnwidth]{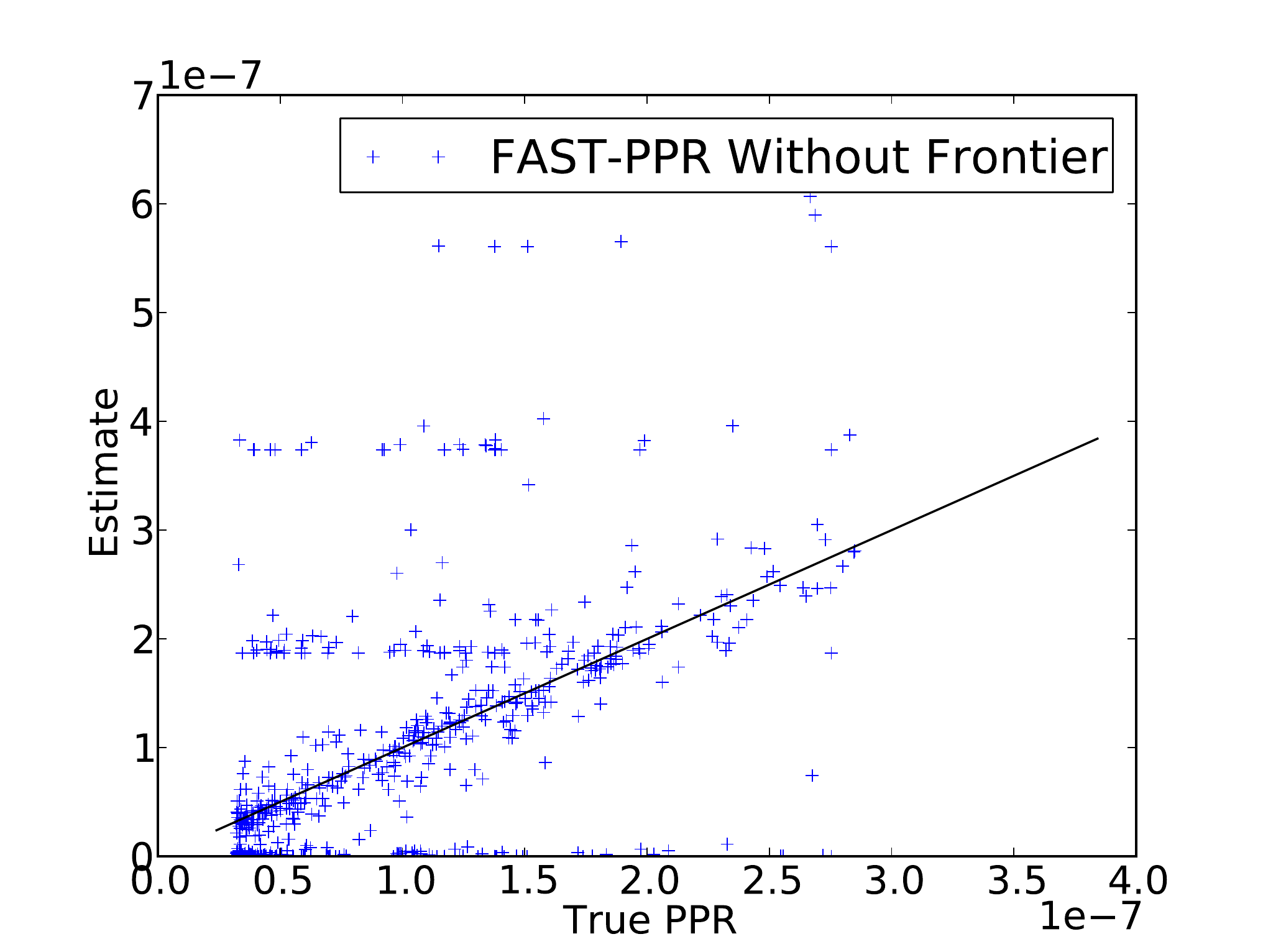}
}
\caption[Accuracy Plots]{
The importance of using the frontier. In each of these plots, a perfect algorithm would place all data points on the line $y=x$. Notice how using inverse-PPR estimates from the target set rather than the frontier results in significantly worse accuracy.}
\label{fig:accuracy}
\end{figure*}

\begin{figure}
\centering
\includegraphics[width=\columnwidth]{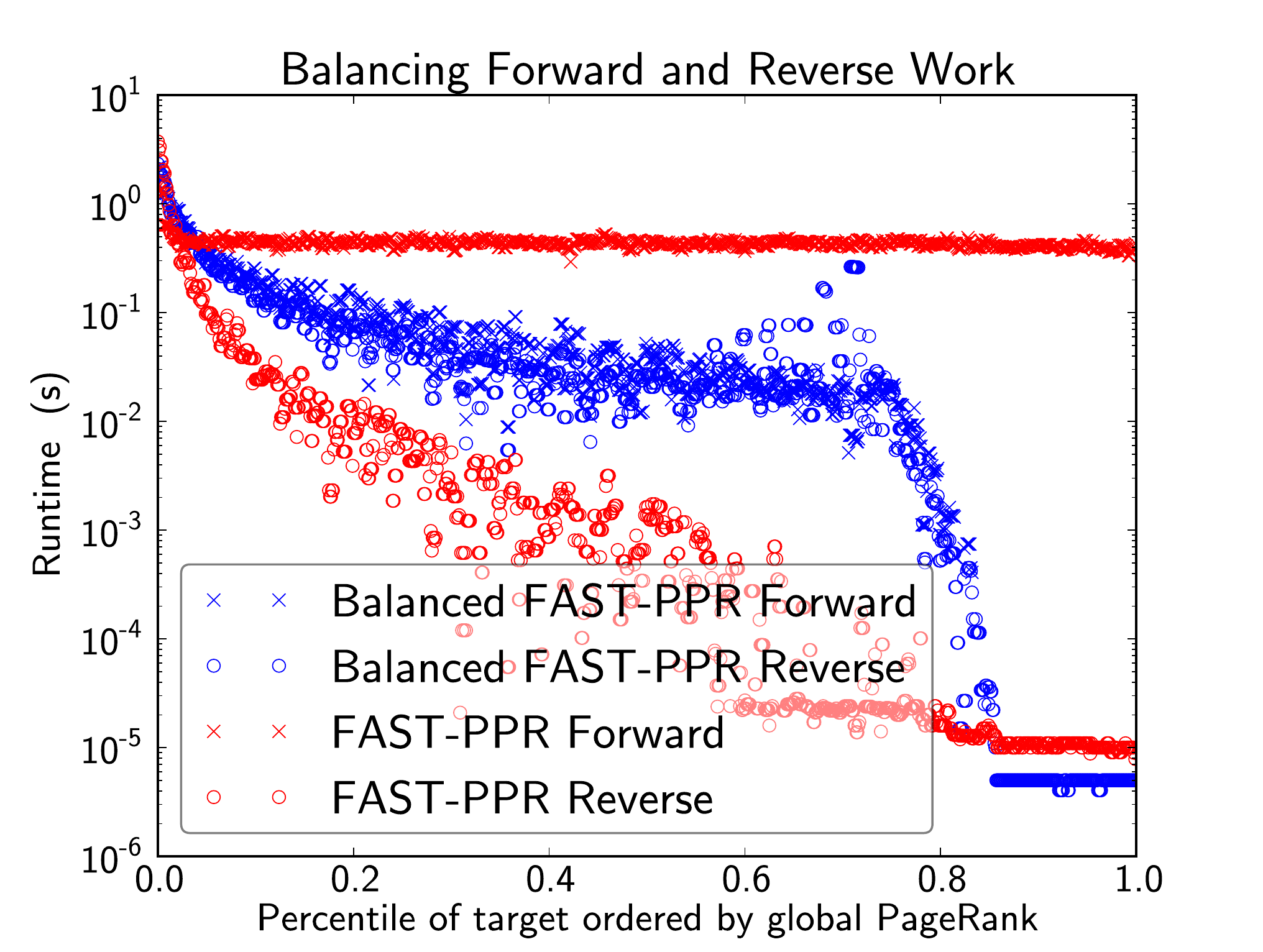}
\caption{Forward and reverse running-times (on log-scale) for FAST-PPR (in red) and BALANCED-FAST-PPR (in blue) as we vary the global PageRank of the target node on the x-axis. Data is for the Twitter-2010 graph and is smoothed using median-of-five smoothing. Notice how there is a significant gap between the forward and backward work in FAST-PPR, and that this gap is corrected by BALANCED-FAST-PPR.}
\label{fig:balance_pagerank}
\end{figure}

\subsection{Some Other Empirical Observations}
\label{ssec:balanceexpt}

\noindent\textbf{$\bullet$ Necessity of the Frontier:} Another question we study experimentally is whether we can modify FAST-PPR to compute the target set $T_t(\e_r)$ and then run Monte-Carlo walks until they hit the target set (rather than the frontier set $F_t(\e_r)$). This may appear natural, as the target set is also a blanket set, and we have good approximations for inverse-PPR values in the target set. Further, using only the target set would reduce the dependence of the running-time on $d$, and also reduce storage requirements for an oracle-based implementation.

It turns out however that \emph{using the frontier is critical} to get good accuracy. Intuitively, this is because nodes in the target set may have high inverse-PPR, which then increases the variance of our estimate. This increase in variance can be visually seen in a scatterplot of the true vs estimated value, as shown in Figure  \ref{fig:accuracy_d} -- note that the estimates generated using the frontier set are much more tightly clustered around the true PPR values, as compared to the estimates generated using the target set.

\noindent\textbf{$\bullet$ Balancing Forward and Reverse Work:} BALANCED-FAST-PPR, as described in Section \ref{ssec:rebalancing}, chooses reverse threshold $\epsilon_r$ dynamically for each target node.  Figure \ref{fig:runtime_0_2} shows that BALANCED-FAST-PPR improves the average running-time across all graphs.

In Figure \ref{fig:balance_pagerank}, we plot the forward and reverse running-times for FAST-PPR and Balanced FAST-PPR as a function of the target PageRank. Note that for high global-PageRank targets, FAST-PPR does too much reverse work, while for low global-PageRank targets, FAST-PPR does too much forward work -- this is corrected by Balanced FAST-PPR.

\section*{Acknowledgements}

Peter Lofgren is supported by a National Defense Science and
Engineering Graduate (NDSEG) Fellowship.

Ashish Goel and Siddhartha Banerjee were supported in part by the
DARPA XDATA program, by the DARPA GRAPHS program via grant FA9550-
12-1-0411 from the U.S. Air Force Office of Scientific Research
(AFOSR) and the Defense Advanced Research Projects Agency (DARPA), and
by NSF Award 0915040.

C. Seshadhri is at Sandia National Laboratories, which is a multi-program laboratory managed and operated by Sandia Corporation, a wholly owned subsidiary of Lockheed Martin Corporation, for the U.S. Department of Energy's National Nuclear Security Administration under contract DE-AC04-94AL85000.

\bibliographystyle{ieeetr}
\bibliography{PPR-refs}

\appendix

\section{FAST-PPR with Approximate Frontier Estimates}
\label{appsec:imperfect_alt}

We now extend to the case where we are given an estimate $\{\widehat{\PR}^{-1}_t(w)\}$ of the inverse-PPR vector of $v$, with maximum additive error $\e_{inv}$, i.e.:
$$\max_{w\in V}\left\{|\PR^{-1}_t(w)-\widehat{\PR}^{-1}_t(w)|\right\}<\e_{inv}.$$
We note here that the estimate we obtain from the APPROX-FRONTIER algorithm (Algorithm \ref{alg:invPPR}) has a \emph{one-sided} error; it always underestimates $\PR_t^{-1}$. However, we prove the result for the more general setting with two-sided error.

First, we have the following lemma:
\begin{lemma}
\label{lem:approxfrontier}
Given $\e_r$ and $\e_{inv}\leq \e_r\beta$ for some $\beta\in(0,1/2)$. Defining approximate target-set $\widehat{T}_t(\e_r)= \{w \in V : \widehat{\PR}^{-1}_t(w)>\e_r\}$, we have:	
\begin{align*}
	T_t((1+\beta)\e_r)&\subseteq\widehat{T}_t(\e_r)\subseteq T_t(\e_r(1-\beta))
\end{align*}
Moreover, let $c_{inv}=\left(\frac{\beta}{1-\beta}\right)<1$; then $\fall w\in\widehat{T}_t(\e_r)$:
\begin{align*}
|\PR^{-1}_t(w)-\widehat{\PR}^{-1}_t(w)|<c_{inv}\PRi_t(w).
\end{align*}
\end{lemma}

\begin{proof}
	The first claim follows immediately from our definition of $T_t$ and the additive error guarantee. Next, from the fact that $\widehat{T}_t(\e_r)\subseteq T_t(\e_r(1-\beta))$. we have that $\fall w\in\widehat{T}_t(\e_r), \PRi_t(w)\geq(1-\beta)\e_r$. Combining this with the additive guarantee, we get the above estimate. 
\end{proof}

Now for FAST-PPR estimation with approximate oracles, the main problem is that if $\e_{inv}=\beta\e_r$, then the additive error guarantee is not useful when applied to nodes in the frontier: for example, we could have $\PR^{-1}_t(w) < \e_{inv}$ for all $w$ in the frontier $\widehat{F}_t(\e_r)$. On the other hand, if we allow the walks to reach the target set, then the variance of our estimate may be too high, since $\PR^{-1}_t(w)$ could be $\Omega(1)$ for nodes in the target set. To see this visually, compare figures \ref{fig:accuracy_c} and \ref{fig:accuracy_d} -- in the first we use the estimate for nodes in the frontier, ending up with a negative bias; in the second, we allow walks to hit the target set, leading to a high variance. 

However, Lemma \ref{lem:approxfrontier} shows that we have a multiplicative guarantee for nodes in $\widehat{T}_t(\e_r)$. We now circumvent the above problems by \emph{re-sampling walks to ensure they do not enter the target set}. We first need to introduce some additional notation. Given a target set $T$, for any node $u\notin T$, we can partition its neighboring nodes $\N^{out}(u)$ into two sets -- $\N_T(u)\triangleq\N^{out}(u)\cap T$ and $\N_{\overline{T}}(u)=\N^{out}(u)\setminus T$. We also define a \emph{target-avoiding random-walk} $RW_T(u)=\{u,V_1,V_2,V_3,\ldots\}$ to be one which starts at $u$, and at each node $w$ goes to a uniform random node in $\N_{\overline{T}}(w)$. 

Now given target-set $T$, suppose we define the \emph{score} of any node $u\notin T$ as $S_T(u)=\frac{\sum_{z\in\N_T(u)}\PR_z(t)}{d^{out}(u)}$. We now have the following lemma:
\begin{lemma}
\label{lem:altestimate}
Let $RW_T(s)=\{u,V_1,V_2,V_3,\ldots\}$ be a target-avoiding random-walk, and $L\sim Geom(\alpha)$. Further, for any node $u\notin T$, let $p_{\overline{T}}(u)=\frac{\N_{\overline{T}}(u)}{d^{out}(u)}$, i.e., the fraction of neighboring nodes of $u$ not in target set $T$. Then:
\begin{align}
\label{eq:altestimate}
\PR_s(t)&=\EE_{RW_T(s),L}\left[\sum_{i=1}^{L}\left(\prod_{j=0}^{i-1}p_{\overline{T}}(V_j)\right)S_T(V_i)\right].
\end{align}
\end{lemma}

\begin{proof}
Given set $T$, and any node $u\notin T$, we can write:
\begin{align*}
\PR_u(t)&=\frac{1-\alpha}{d^{out}(u)}\left(\sum_{z\in\N_T(u)}\PR_z(t) + \sum_{z\in\N_{\overline{T}}(u)} \PR_z(t)\right).
\end{align*}
Recall we define the score of node $u$ given $T$ as $S_T(u)=\frac{\sum_{z\in\N_T(u)}\PR_z(t)}{d^{out}(u)}$, and the fraction of neighboring nodes of $u$ not in $T$ as $p_{\overline{T}}(u)=\frac{\N_{\overline{T}}(u)}{d^{out}(u)}$. Suppose we now define the \emph{residue} of node $u$ as $R_T(u)=\frac{\sum_{z\in\N_{\overline{T}}(u)}\PR_z(t)}{|\N_{\overline{T}}(u)|}$. Then we can rewrite the above formula as:
\begin{align*}
\PR_u(t)&=(1-\alpha)\left(S_T(u)+p_{\overline{T}}(u)R_T(u)\right).
\end{align*}
Further, from the above definition of the residue, observe that $R_T(u)$ is the \emph{expected value of $\PR_W(t)$ for a uniformly picked node $W\in\N_{\overline{T}}(u)$}, i.e., a uniformly chosen neighbor of $u$ which is not in $T$. Recall further that we define a target-avoiding random-walk $RW_T(s)=\{s,V_1,V_2,V_3,\ldots\}$ as one which at each node picks a uniform neighbor not in $T$. Thus, by iteratively expanding the residual, we have:
\begin{align*}
\PR_s(t)&=\EE_{RW_T(s)}\left[\sum_{i=1}^{\infty}(1-\alpha)^i\left(\prod_{j=0}^{i-1}p_{\overline{T}}(V_j)\right)S_T(V_i)\right]\\
&=\EE_{RW_T(s),L}\left[\sum_{i=1}^{L}\left(\prod_{j=0}^{i-1}p_{\overline{T}}(V_j)\right)S_T(V_i)\right].
\end{align*}
\end{proof}

\begin{algorithm}[!t]
\caption{Theoretical-FAST-PPR$(s,t,\delta)$}
\label{alg:FASTPPR_approx}
\begin{algorithmic}[1] 
\REQUIRE start node $s$, target node $t$, threshold $\delta$, reverse threshold $\e_r$, relative error $c$, failure probability $p_{fail}$.
\STATE Define $\e_f = \frac{\delta}{\e_r}$, $\beta=\frac{c}{3+c}$, $p_{min}=\frac{c}{3(1+\beta)}$
\STATE Call FRONTIER$(t,\e_r)$ (with $\beta$ as above) to obtain sets $\widehat{T}_t(\e_r), \widehat{F}_t(\e_r)$, inverse-PPR values $(\widehat{\PR}^{-1}_t(w))_{w\in T_t(\e_r)}$ and target-set entering probabilities $p_{\overline{T}}(u)\fall u\in \widehat{F}_t(\e_r)$.
\STATE Set $l_{max}=\log_{1-\alpha}\left(\frac{c\delta}{3}\right)$ and $n_f=\frac{45l_{max}}{c^2\e_f}\log\left(\frac{2}{p_{fail}}\right)$ 
\FOR{index $i\in [n_f]$}
\STATE Generate $L_i\sim Geom(\alpha)$
\STATE Setting $T=\widehat{T}_t(\e_r)$, generate target-avoiding random-walk $RW_T^i(s)$ via rejection-sampling of neighboring nodes.
\STATE Stop the walk $RW_T^i(s)$ when it encounters any node $u$ with $p_{\overline{T}}(u)<p_{min}$, OR,  when the number of steps is equal to $\min\{L_i,l_{max}\}$. 
\STATE Compute walk-score $S_i$ as the weighted score of nodes on the walk, according to equation \ref{eq:altestimate}.
\ENDFOR 
\RETURN $\frac{1}{n_f}\sum_{i=1}^{n_f}S_i$
\end{algorithmic}
\end{algorithm}

We can now use the above lemma to get a modified FAST-PPR algorithm, as follows: First, given target $t$, we find an approximate target-set $T=\widehat{T}(\e_r)$. Next, from source $s$, we generate a target-avoiding random walk, and calculate the weighted sum of scores of nodes. The walk collects the entire score the first time it hits a node in the frontier; subsequently, for each node in the frontier that it visits, it accumulates a smaller score due to the factor $p_{\overline{T}}$. Note that $p_{\overline{T}}(u)$ is also the probability that a random neighbor of $u$ is not in the target set. Thus, at any node $u$, we need to sample on average $1/p_{\overline{T}}(u)$ neighbors until we find one which is not in $T$. This gives us an efficient way to generate a target-avoiding random-walk using rejection sampling -- we stop the random-walk any time the current node $u$ has $p_{\overline{T}}(u)<p_{min}$ for some chosen constant $p_{min}$, else we sample neighbors of $u$ uniformly at random till we find one not in $T$. The complete algorithm is given in Algorithm \ref{alg:FASTPPR_approx}. 

Before stating our main results, we briefly summarize the relevant parameters. Theoretical-FAST-PPR takes as input a pair of node $(s,t)$, a threshold $\delta$, desired relative error $c$ and desired probability of error $p_{fail}$. In addition, it involves six internal parameters -- $\e_r,\e_f,\beta,p_{min},l_{max}$ and $n_f$ -- which are set based on the input parameters $(\delta,p_{fail},c)$ (as specified in Algorithm \ref{alg:FASTPPR_approx}). $\e_r$ and $\beta$ determine the accuracy of the PPR estimates from FRONTIER; $n_f$ is the number of random-walks; $l_{max}$ and $p_{min}$ help control the target-avoiding walks. 

We now have the following guarantee of correctness of Algorithm \ref{alg:FASTPPR_approx}. For ease of exposition, we state the result in terms of relative error of the estimate -- however it can be adapted to a classification error estimate in a straightforward manner.
\begin{theorem}
\label{thm:fastpprapprox}
Given nodes $(s,t)$ such that $\PR_s(t)>\delta$, desired failure probability $p_{fail}$ and desired tolerance $c<1$. Suppose we choose parameters as specified in Algorithm \ref{alg:FASTPPR_approx}. Then with probability at least $1-p_{fail}$, the estimate returned by Theoretical-FAST-PPR satisfies:
\begin{align*}
\left|\PR_s(t)-\frac{1}{n_f}\sum_{i=1}^{n_f}S_i\right|<c\PR_s(t).
\end{align*}
\end{theorem}

Moreover, we also have the following estimate for the running-time:
\begin{theorem}
\label{thm:perffastppr_approx}
Given threshold $\delta$, teleport probability $\alpha$, desired relative-error $c$ and desired failure probability $p_{f}$, the Theoretical-FAST-PPR algorithm with parameters chosen as in Algorithm \ref{alg:FASTPPR_approx}, requires:
\begin{align*}
\bullet\mbox{ Average reverse-time}&=O\left(\frac{d}{\alpha\beta\e_r}\right)=O\left(\frac{d}{c\alpha\e_r}\right)\\
\bullet\mbox{ Average forward-time}&\leq O\left(\frac{n_f}{\alpha p_{min}}\right)\\
&=O\left(\frac{\log(1/p_{f})\log(1/\delta)}{c^3\alpha\e_f\log(1/(1-\alpha))}\right),
\end{align*}
where $d=\frac{m}{n}$ is the average in-degree of the network.
\end{theorem}

We note that in Theorem \ref{thm:perffastppr_approx}, the reverse-time bound is averaged over random $(s,t)$ pairs, while the forward-time bound is averaged over the randomness in the Theoretical-FAST-PPR algorithm (in particular, for generating the target-avoiding random-walks). As before, we can get worst-case runtime bounds by storing the frontier estimates for all nodes. Also, similar to Corollary \ref{corr:balance}, we can balance forward and reverse times aas follows:

\begin{corollary}
\label{corr:apprbalance}
Let $\e_r=\sqrt{\frac{\delta c^2d\log(1/(1-\alpha))}{\log(1/p_{fail})\log(1/\delta)}}$, where $d=m/n$, and $\e_f=\delta/\e_r$. Then Theoretical-FAST-PPR has an average running-time of $O\left(\frac{1}{\alpha c^2}\sqrt{\frac{d}{\delta}}\sqrt{\frac{\log(1/p_{fail})\log(1/\delta)}{\log(1/(1-\alpha))}}\right)$ per-query, or alternately, worst-case running time and per-node pre-computation and storage of $O\left(\frac{1}{\alpha c^2}\sqrt{\frac{d}{\delta}}\sqrt{\frac{\log(1/p_{fail})\log(1/\delta)}{\log(1/(1-\alpha))}}\right)$. 
\end{corollary}

\begin{proof}[of theorem \ref{thm:perffastppr_approx}]
The reverse-time bound is same as for Theorem \ref{thm:perffastppr}. For the forward-time, observe that each target-avoiding random-walk $RW_T(s)$ takes \emph{less than} $1/\alpha=O(1)$ steps on average (since the walk can be stopped before $L$, either when it hits $l_{max}$ steps, or encounters a node $u$ with $p_{\overline{T}}(u)<p_{min}$). Further, for each node on the walk, we need (on average) at most $1/p_{min}$ samples to discover a neighboring node $\notin T$.
\end{proof}

\begin{proof}[of theorem \ref{thm:fastpprapprox}]
Define target set $T=\widehat{T}_t(\e_r)$. Using Lemma \ref{lem:altestimate}, given a target-avoiding random-walk $RW_T(s)=\{s,V_1,V_2,\ldots\}$, we define our estimator for $\PR_s(t)$ as: 
\begin{align*}
\widehat{\PR}_s(t)&=\sum_{i=1}^{L}\left(\prod_{j=0}^{i-1}p_{\overline{T}}(V_j)\right)\widehat{S}_T(V_i).
\end{align*}
Now, from Lemma \ref{lem:approxfrontier}, we have that for any node $u\notin T$, the \emph{approximate score} $\widehat{S}_T(u)=\frac{1}{d^{out}(u)}\sum_{z\in\N_{T}(u)}\widehat{\PR}_z(t)$ lies in $(1\pm c_{inv})S_T(u)$. Thus, for any $s\notin \widehat{T}_t(\e_r)$, we have:
\begin{align*}
\left|\EE_{RW_T}\left[\widehat{\PR}_s(t)\right]-\PR_s(t)\right|\leq c_{inv}\PR_s(t)
\end{align*}	 
with $c_{inv}=\frac{\beta}{1-\beta}=\frac{c}{3}$ (as we chose $\beta=\frac{c}{3+c}$).

Next, suppose we define our estimate as $S=\frac{1}{n_f}\sum_{i=1}^{n_f}S_i$. Then, from the triangle inequality we have:	
\begin{align}
\label{eq:triangle}
\left|S-\PR_s(t)\right|\leq&
\left|S-\EE[S]\right|+\left|\EE\left[S\right]-\EE\left[\widehat{\PR}_s(t)\right]\right|+\\&\left|\EE_{RW_T}\left[\widehat{\PR}_s(t)\right]-\PR_s(t)\right|.
\end{align}	 

We have already shown that the third term is bounded by $c\PR_s(t)/3$. The second error term is caused due to two mutually exclusive events -- stopping the walk due to truncation, or due to the current node having $p_{\overline{T}}(u)$ less than $p_{min}$. To bound the first, we can re-write our estimate as:
\begin{align*}
&\widehat{\PR}_u(t)=\sum_{i=1}^{L}\left(\prod_{j=0}^{i-1}p_{\overline{T}}(V_j)\right)\widehat{S}_T(V_i)\\
&\leq\EE_{RW_T}\left[\sum_{i=1}^{L\wedge l_{max}}\left(\prod_{j=0}^{i-1}p_{\overline{T}}(V_j)\right)S_T(V_i)\right]+(1-\alpha)^{(l_{max}+1)}, 
\end{align*}
where $L\sim Geom(\alpha)$, and $L\wedge l_{max}=\min\{L,l_{max}\}$ for any $l_{max}>0$. Choosing $l_{max}=\log_{1-\alpha}(c\delta/3)$, we incur an additive loss in truncation which is at most $c\delta/3$ -- thus for any pair $(s,t)$ such that $\PR_s(t)>\delta$, the loss is at most $c\PR_s(t)/3$. On the other hand, if we stop the walk at any node $u$ such that $p_{\overline{T}}(u)\leq p_{min}$, then we lose at most a $p_{min}$ fraction of the walk-score. Again, if $\PR_s(t)>\delta$, then we have from before that $\widehat{\PR}_s(t)<\delta(1+\beta)$ -- choosing $p_{min}=\frac{c}{3(1+\beta)}$, we again get an additive loss of at most $c\PR_s(t)/3$.

Finally, to show a concentration for $S$, we need an upper bound on the estimates $S_i$. For this, note first that for any node $u$, we have $\widehat{S}_T(u)\leq\e_r$. Since we truncate all walks at length $l_{max}$, the per-walk estimates $S_i$ lie in $[0,l_{max}\e_r]$. Suppose we define $T_i=\frac{S_i}{l_{max}\e_r}$ and $T=\sum_{i\in[n_f]}T_i$; then we have $T_i\in[0,1]$. Also, from the first two terms in equation \ref{eq:triangle}, we have that $|\EE[S_i]-\PR_s(t)|<2c\PR_s(t)/3$, and thus $\left(1-\frac{2c}{3}\right)\frac{n_f\PR_s(t)}{l_{max}\e_r}\leq\EE[T]\leq\left(1+\frac{2c}{3}\right)\frac{n_f\PR_s(t)}{l_{max}\e_r}$. Now, as in Theorem \ref{thm:fastpprmain}, we can now apply standard Chernoff bounds to get:
\begin{align*}
\PP\left[\left|S-\EE[S]\right|\geq \frac{c\PR_s(t)}{3}\right]
&= \PP\left[\left|T-\EE[T]\right|\geq \frac{cn_f\PR_s(t)}{3l_{max}\e_r}\right]\\
&\leq \PP\left[\left|T-\EE[T]\right|\geq \frac{c\EE[T]}{3(1+2c/3)}\right]\\
&\leq 2\exp{\left(-\frac{\left(\frac{c}{3+2c}\right)^2\EE[T]}{3}\right)}\\
&\leq 2\exp{\left(-\frac{c^2(1-2c/3)n_f\PR_s(t)}{3(3+2c)^2l_{max}\e_r}\right)}\\
&\leq 2\exp{\left(-\frac{c^2(3-2c)n_f\e_f}{9(3+2c)^2l_{max}}\right)}.
\end{align*}
Since $c\in[0,1]$, setting $n_f=\frac{45}{c^2}\cdot\frac{l_{max}\log\left(2/p_{fail}\right)}{\e_f}$ gives the desired failure probability.
\end{proof}

\section{Details of Lower Bound}
\label{appsec:lower}

\begin{proof}[of Theorem~\ref{thm:lb}]
We perform a direct reduction. Set $N = \lfloor 1/10\delta \rfloor$, and consider the distributions $\cG_1$ and $\cG_2$.
We will construct a distinguisher using a single query to an algorithm for \signif$(\delta)$.
Hence, if there is an algorithm for \signif$(\delta)$ taking less than $1/100\sqrt{\delta}$ queries, then Theorem~\ref{thm:GR}
will be violated. 
We construct a distinguisher as follows. Given graph $G$, it picks two uniform random nodes $s$ and $t$.
We run the algorithm for \signif$(\delta)$. If it accepts (so it declares $\PR_s(t) > \delta$),
then the distinguisher outputs $\cG_1$ as the distribution that $G$ came from. Otherwise, it 
outputs $\cG_2$.

We prove that the distinguisher is correct with probability $> 2/3$. 
First, some preliminary lemmas.

\begin{lemma} \label{lem:g1} With probability $> 3/4$ over the choice of $G$ from $\cG_1$, for any nodes $s,t$, $\PR_s(t) > \delta$.
\end{lemma}

\begin{proof} A graph formed by picking $3$ uniform random matchings is an \emph{expander} with probability
$1 - \exp(-\Omega(N)) > 3/4$ (Theorem 5.6 of~\cite{motwani2010randomized}). Suppose $G$ was an expander. 
Then a random walk of length $10\log N$ from $s$ in $G$ is guaranteed
to converge to the uniform distribution. Formally, let $W$ be the random walk matrix of $G$.
For any $\ell \geq \flo{10\log N}$, $\|W^\ell \be_s - \bone/N\|_2 \leq 1/N^2$, where $\bone$ is the all ones vector~\cite{motwani2010randomized}.
So for any such $\ell$, $(W^\ell \be_s)(t) \geq 1/2N$.
By standard expansions, $\PR_s = \sum_{\ell=0}^\infty \alpha (1-\alpha)^\ell W^\ell \be_s$.
We can bound 

\begin{eqnarray*}
\PR_s(t) & \geq & \sum_{\ell \geq \flo{10\log N}} \alpha (1-\alpha)^\ell  (W^\ell \be_s)(t) \\
& \geq & (2N)^{-1}\sum_{\ell \geq \flo{10\log N}} \alpha (1-\alpha)^\ell \\
& = & (1-\alpha)^{\flo{10\log N}} (2N)^{-1}
\end{eqnarray*}

Setting $\alpha = 1/100\log(1/\delta)$ and $N = \flo{1/10\delta}$,
we get $\PR_s(t) > 1/6N \geq \delta$. All in all, when $G$
is an expander, $\PR_s(t) > \delta$.
\end{proof}

\begin{lemma} \label{lem:g2} Fix any $G$ from $\cG_2$. 
For two uniform random nodes $s,t$ in $G$, $\PR_s(t) = 0$ with probability at least $3/4$.
\end{lemma}

\begin{proof}
Any graph in $G$ has $4$ connected components, each of size $N/4$.
The probability that $s$ and $t$ lie in different components is $3/4$, in which case $\PR_s(t) = 0$.
\end{proof}

If the adversary chose $\cG_1$, then $\PR_s(t) > \delta$. If he chose $\cG_2$,
$\PR_s(t) = 0$ -- by the above lemmas, each occurs with probability $>3/4$ . In either case, the probability that \signif$(\delta)$ errs is at most $1/10$. By the union bound, the distinguisher is correct overall with probability at least $2/3$.
\end{proof}

\section{BALANCED-FAST-PPR Pseudocode}
\label{appsec:balanced_code}

\begin{algorithm}[!t]
\caption{BALANCED-FAST-PPR$(s,t,\delta)$}
\label{alg:BALANCED_FASTPPR}
\begin{algorithmic}[1] 
\REQUIRE Graph $G$, teleport probability $\alpha$, start node $s$, target node $t$, threshold $\delta$
\STATE Set accuracy parameters $c$, $\beta$  (These trade-off speed and accuracy -- in experiments we use $c = 350$, $\beta=1/6$.)
\STATE Call BALANCED-FRONTIER$(G, \alpha, t, c, \beta)$ to obtain target set $T_t(\e_r)$, frontier set $ F_t(\e_r)$, inverse-PPR values $(\PR^{-1}_t(w))_{w\in F_t(\e_r) \cup T_t(\e_r)}$, and reverse threshold $\e_r$
\STATE (The rest of BALANCED-FAST-PPR is identical to FAST-PPR (Algorithm \ref{alg:FASTPPR}) from line 3)
\end{algorithmic}
\end{algorithm}

\begin{algorithm}[!t]
\caption{BALANCED-FRONTIER$(t,c,\beta)$}
\label{alg:balancedInvPPR}
\begin{algorithmic}[1]
\REQUIRE Graph $G=(V,E)$, teleport probability $\alpha$, target node $t$, accuracy factors $c,\beta$.
\STATE Set constant $T_{walk}$ to be the (empirical) average time it takes to generate a walk.
\STATE Initialize (sparse) estimate-vector $\widehat{\PR}^{-1}_t$ and (sparse) residual-vector $r_t$ as: $\begin{cases} \widehat{\PR}^{-1}_t(u)=r_t(u)=0 \text{ if } u \neq t \\ \widehat{\PR}^{-1}_t(t)=r_t(t)=\alpha \end{cases}$
\STATE Initialize target-set $\widehat{T}_t=\{t\}$, frontier-set $\widehat{F}_t=\{\}$ 
\STATE Initialize $\epsilon_r = 1/\beta$.
\STATE Start a timer to track time spent.
\STATE Define function FORWARD-TIME($\epsilon_r$) = $T_{walk} \cdot k$ where $k=c \cdot \epsilon_r / \delta$ is the number of walks needed.

\WHILE{time spent < FORWARD-TIME($\epsilon_r$)} 

\STATE $w = \argmax_u r_t(u)$
\FOR{$u\in\mathcal{N}^{in}(w)$}
\STATE      $\Delta =  (1-\alpha).\frac{r_t(w)}{d^{out}(u)}$
\STATE      $\widehat{\PR}^{-1}_t(u) = \widehat{\PR}^{-1}_t(u) + \Delta, r_t(u) = r_t(u) + \Delta$
\IF{$\widehat{\PR}^{-1}_t(u)>\e_r$}
\STATE $\widehat{T}_t=\widehat{T}_t \cup \{u\}\,,\, \widehat{F}_t=\widehat{F}_t \cup \mathcal{N}^{in}(u)$
\ENDIF
\ENDFOR
\STATE Update $r_t(w)=0, \epsilon_r = (\max_u r_t(u)) / (\alpha \cdot \beta)$
\ENDWHILE
\STATE $\widehat{F}_t=\widehat{F}_t\setminus \widehat{T}_t$
\RETURN $\widehat{T}_t,\widehat{F}_t, (\widehat{\PR}^{-1}_t(w))_{w\in \widehat{F}_t \cup \widehat{T}_t}, \epsilon_r$
\end{algorithmic}
\end{algorithm}    

We introduced BALANCED-FAST-PPR in Section \ref{ssec:rebalancing}, where for brevity, we only gave a brief description. We provide the complete pseudocode for the Balanced FAST-PPR algorithm in Algorithms \ref{alg:BALANCED_FASTPPR} and \ref{alg:balancedInvPPR}.
 
We note that similar to Algorithm \ref{alg:invPPR}, the update rule for inverse-PPR estimates in Algorithm \ref{alg:balancedInvPPR} is based on local-update algorithms from \cite{Andersen2007,Lofgren2013} -- the novelty here is in using runtime estimates for random-walks to determine the amount of backward work (i.e., up to what accuracy the local-update step should be executed). Also note that in Algorithm \ref{alg:balancedInvPPR}, we need to track the maximum residual estimate -- this can be done efficiently using a binary heap. For more details, please refer our source code, which is available at:
\url{http://cs.stanford.edu/~plofgren/fast_ppr/}

\section{Runtime vs. less-Accurate Monte-Carlo}
\label{appsec:expts}

\begin{figure}
\centering
\subfigure[Sampling targets uniformly]{
\label{fig:runtime_0_2_a_mc1}
\includegraphics[width=1\columnwidth]{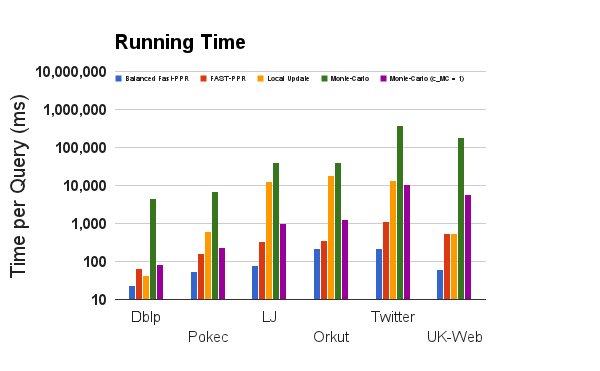}
}

\subfigure[Sampling targets from PageRank distribution]{
\label{fig:runtime_0_2_b_mc1}
\includegraphics[width=1\columnwidth]{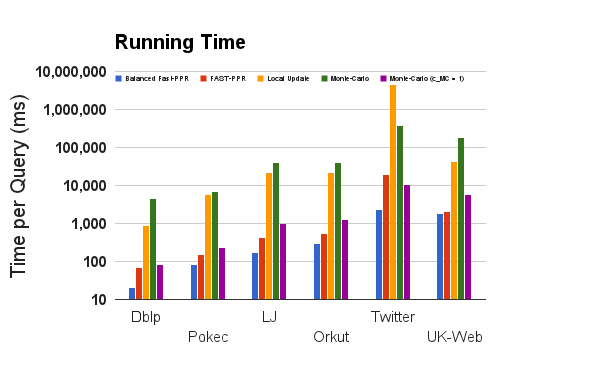}
}
\caption[Running-time Plots]{Average running-time (on log-scale) for different networks, with $\delta=\frac{4}{n}, \alpha=0.2$, for $1000$ $(s,t)$ pairs. We compare to a less-accurate version of Monte-Carlo, with only $\frac{1}{\epsilon_f}$ walks -- notice that Balanced FAST-PPR is still faster than Monte-Carlo.}
\label{fig:runtime_0_2_mc1}
\end{figure}

In Figure \ref{fig:runtime_0_2_mc1} we show the data data as in Figure \ref{fig:runtime_0_2} with the addition of a less accurate version of Monte-Carlo. When using Monte-Carlo to detect an event with probability $\epsilon_f$, the minimum number of walks needed to see the event once on average is $\frac{1}{\epsilon_f}$. The average relative error of this method is significantly greater than the average relative error of FAST-PPR, but notice that even with this minimal number of walks, Balanced FAST-PPR is faster than Monte-Carlo.

\end{document}